\documentclass[11pt]{article}

\usepackage[a4paper]{geometry}

\usepackage[english]{babel}
\usepackage[utf8]{inputenc}
\usepackage[T1]{fontenc}
\usepackage{epsfig, float}
\usepackage{appendix}

\usepackage{longtable}
\usepackage{mathcomp}
\usepackage{dsfont}
\usepackage{lmodern}
\usepackage{scrextend}
\usepackage{color}

\usepackage[nottoc]{tocbibind}

\usepackage{amsmath,amssymb,amsfonts,amsthm}
\usepackage{amsfonts}
\usepackage{url}
\usepackage{bbm}
\usepackage{mathrsfs} 

\usepackage{enumerate}
\usepackage[hidelinks]{hyperref}

\newcommand{\norm}[1]{\left\lVert#1\right\rVert}

\newcommand{\R}{\mathbb{R}}
\newcommand{\C}{\mathbb{C}}
\newcommand{\N}{\mathbb{N}}

\renewcommand{\Re}{\operatorname{Re}}
\renewcommand{\Im}{\operatorname{Im}}

\newtheorem{theorem}{Theorem}[section]
\newtheorem{lemma}[theorem]{Lemma}
\newtheorem{proposition}[theorem]{Proposition}
\newtheorem{remark}[theorem]{Remark}
\newtheorem{definition}[theorem]{Definition}

\numberwithin{equation}{section} 

\usepackage{cancel}

\usepackage[normalem]{ulem}
\usepackage{todonotes}

\definecolor{miguelscolor}{rgb}{.7,.2,.2}
\definecolor{dirkscolor}{rgb}{.2,.2,.7}
\definecolor{felixscolor}{rgb}{.2,.7,.2}

\newcommand{\ac}[1]{}

\date{\today}

\title{\Large{\textsc{Relation between\\the Resonance and the Scattering
Matrix\\in the massless Spin-Boson Model}}}

\author{Miguel Ballesteros\thanks{\texttt{miguel.ballesteros@iimas.unam.mx},
    Instituto de Investigaciones en Matem\'aticas Aplicadas y en Sistemas,
    Universidad Nacional Aut\'anoma de M\'exico}, Dirk-Andr\'e
    Deckert\thanks{\texttt{deckert@math.lmu.de}, Mathematisches Institut
    der Ludwig-Maximilians-Universität München}, Felix
    H\"anle\thanks{\texttt{haenle@math.lmu.de}, Mathematisches Institut
der Ludwig-Maximilians-Universität München}}

\begin{document}
\maketitle

\begin{abstract}
We establish the precise relation between the integral kernel of the scattering
matrix and the resonance in the massless Spin-Boson model which describes the
interaction of a two-level quantum system with a second-quantized scalar field.
For this purpose, we derive an explicit formula for the two-body scattering
matrix.  We impose an ultraviolet cut-off and assume a slightly
less singular behavior of the boson form factor of the relativistic scalar
field but no infrared cut-off.    The purpose of this
work is to bring together scattering
and resonance theory and 
arrive at a similar result as  provided by  Simon in \cite{simonnbody},
 where  it was shown that the singularities of the
meromorphic continuation of the integral kernel of the scattering matrix
are located precisely at the resonance energies.
The corresponding  problem has been open
 in   quantum field theory  ever  since.
  To the best of our knowledge, the presented formula provides the
first rigorous connection between resonance and scattering theory in the sense
of \cite{simonnbody} in a model of quantum field theory. 
\end{abstract}

\section{Introduction}
\label{sec:introduction}

In this paper, we  analyze the massless Spin-Boson model which is a
 non-trivial model of quantum field theory.
It can be seen as a model of a two-level atom interacting with its
second-quantized scalar field, and hence, provides a widely employed
 model  for quantum optics which gives insights into scattering processes
between photons and atoms.  The unperturbed energies of the two-level atom
shall be denoted by  real numbers $0 = e_0<e_1$.  It is
well-known that after switching on the interaction with a massless scalar
field, which may induce transitions between the atom levels, the free ground
state energy $e_0$ is shifted to the interacting ground state energy
$\lambda_0$ while the free excited state with energy $e_1$ turns into a
resonance with complex ``energy'' $\lambda_1$. 

One of the main  mathematical difficulties in the study of the massless
Spin-Boson model is the absence of a spectral gap which does not allow a
straight-forward application of  regular perturbation theory.  Several
techniques have been developed to overcome this difficulty.
There are  two methods that rigorously
address this problem: The so-called
 renormalization group (see e.g.\
\cite{bfs1,bfs2,bfs3,bcfs,bfs100,bbf,fgs100,feshbach,s,f,bffs}) which
was  the first one  used to construct resonances in models of quantum field theory, and
 furthermore,  the so-called multiscale method which was developed in
\cite{pizzo1,pizzo2,bach,bbp} and also successfully applied in various models
of quantum field theory.  In both cases, a family of  spectrally dilated
Hamiltonians is analyzed since this allows for complex eigenvalues.  Our work
draws from the results obtained in a previous article \cite{bdh-res}  which is
build on the latter technique mentioned above.  Beyond  the
construction, we obtained several spectral estimates and analyticity properties
 in \cite{bdh-res}  which are crucial ingredients  for this work.

In addition to the resonance theory, also the scattering theory is
well-established in various models of quantum field theory, e.g., in
\cite{fau1,fau2,fau3,fgs1,fgs2}, and in particular in the massless Spin-Boson
model, e.g., in  \cite{rgk,rk,rk2,derezinski,bkz}.  The purpose of this work is
to bring these two well-developed fields together and to arrive at a similar
result as  provided by  Simon in \cite{simonnbody}.
Therein,  it was shown that
the singularities of the meromorphic continuation of the integral kernel of the
scattering matrix are located precisely at the resonance energies.
 To the best of our knowledge, this question has not yet been 
addressed in models of quantum field theory, which is most probably due to the
fact that quantum field models involve new subtleties as compared to the quantum
mechanical ones. These can however be addressed with the recently developed methods of
multiscale analysis and spectral renormalization (while we rely on the former in this work).  
 We provide a representation of the scattering
matrix in terms of an expectation value of the resolvent of a spectrally
dilated Hamiltonian; see Theorem  \ref{FKcor}  below. The relation of the scattering
matrix and the resonance can then be read of this formula; see  Eqs.\ \eqref{scatteringformulapp} and \eqref{scatteringformulapp1}   below. Loosely put, our results imply that, for the photon momenta $|k'|$ in a neighborhood of $\Re\lambda_1-\lambda_0$, the
leading order (in $g$ for small $g$) of the  integral kernel of the transition matrix fulfills 
\begin{align}
 |T(k,k')|^2\sim 
  \frac{ E_1^2 g^4 }{ ( |k'| + \lambda_0   
 - \Re \lambda_1 )^2
+ g^4 E_1^2},
\label{eq:lorentzian}
\end{align} 
where we define 
\begin{align}
\label{def:e1}
 E_1 :=  g^{-2} \Im \lambda_1 ,
\end{align}
 and it turns out that there are constant numbers $E_I<0$,  $a>0$  and a uniformly bounded function $\Delta\equiv \Delta(g)$ such that $E_1=E_I+g^{ a }  \Delta$.
Heuristically, for an experiment in which a two-level atom is irradiated with
monochromatic incoming light quanta of  momentum $k'\in\R^3$, the relation
\eqref{eq:lorentzian} states that the intensity of  the outgoing light quanta
with momentum $k\in\R^3$ is proportional to  $ |T(k,k')|^2$, which is given as
a Lorentzian function with  maximum at $|k'|=\Re \lambda_1 -\lambda_0$  and
width $2\Im \lambda_1$. This relation is already found as folklore knowledge
in physics text-books.  In this work we give a rigorous derivation in the
model at hand.   On the other hand, the relation between the imaginary value of the
resonance and the decay rate of the unstable excited state  was established
rigorously in several articles \cite{sigal,hasler,sf,bmw}.

In \cite{bfp}, a rigorous  mathematical justification of   Bohr’s frequency condition  is derived, using  an expansion of the scattering amplitudes with respect to powers the finestructur constant for the Pauli-Fierz model. In particular, they calculate the leading order term  and provide an algorithm for computing the other terms.  In \cite{bkz}, the photoelectric effect is studied for a model of an atom with a single bound state,  coupled to the quantized electromagnetic  field.  In their work, they use similar techniques  for estimating time evolution as the ones presented in this manuscript.

\subsection{The Spin-Boson model}
\label{sec:defmodel}
In  this section  we introduce the considered model and
 state preliminary definitions and well-known tools and facts from which we
start our analysis. If the reader is already familiar with the introductory
Sections 1.1 until 1.2 of \cite{bdh-res}, these subsections can be skipped. The
notation is identical and these subsections are only given for the purpose of
self-containedness. \\

The non-interacting Spin-Boson Hamiltonian is defined as
\begin{align}
\label{h0def}
H_0:=K + H_f , \qquad K:= \begin{pmatrix}
e_1 & 0 \\
0 & e_0
\end{pmatrix} ,
\qquad
H_f:=\int \mathrm{d^3}k \, \omega(k) a(k)^* a(k).
\end{align}
We regard $K$  as an idealized free Hamiltonian of a two-level atom. As already
stated in the introduction, its two energy levels are denoted by the real
numbers $0 = e_0 <e_1$. Furthermore, $H_f$ denotes the free
Hamiltonian of a massless scalar field having dispersion relation
$\omega(k)=|k|$, and $a,a^*$ are the annihilation and creation operators on the
standard Fock space which will be defined  in  \eqref{def:a} and \eqref{def:a1}  below. In the following we will
sometimes call $K$ the atomic part, and $H_f$ the free field part of the
Hamiltonian.  The sum of the free two-level atom Hamiltonian $K$ and the free
field Hamiltonian $H_f$ will simply be referred to as the ``free Hamiltonian''
$H_0$.  The interaction term reads
\begin{align}
\label{interaction}
V:= \sigma_1\otimes \left( a(f) + a(f)^*\right) , \qquad   \sigma_1:= \begin{pmatrix}
0 & 1 \\
1 & 0
\end{pmatrix} , 
\end{align}
where the boson form factor  is given by
\begin{align}
f: \R^3 \setminus \{0\}\to \R , \qquad k\mapsto e^{-\frac{k^2}{\Lambda^2}}|k|^{-\frac{1}{2}+\mu} .
\label{eq:f}
\end{align}
Note that the relativistic form factor of a scalar field should rather be $ f(k)=(2\pi)^{-\frac{3}{2}}(2|k|)^{-\frac{1}{2}} $, which however renders the model ill-defined due to the fact that such an $f$ would not be square integrable. This is referred to as  ultraviolet divergence.  In our case, the gaussian factor in \eqref{eq:f} acts as an ultraviolet cut-off for $\Lambda>0$ being the ultraviolet cut-off parameter and in addition the fixed number 
\begin{align}
\label{const:mu}
\mu\in (0,1/2)
\end{align}
 implies a regularization of the infrared singularity at $k=0$ which is a
 technical assumption chosen for this work to keep the proofs more tractable.
 With  additional work, one can also treat the case $\mu=0$ with methods
 described in \cite{bbkm}. The missing factor of
 $2^{-\frac{1}{2}}(2\pi)^{-\frac{3}{2}}$  will be absorbed in the coupling
 constant  $g$  in our notation. Note that the form factor $f$ only depends on the
 radial part of $k$. To emphasize this,  we often write $f(k)\equiv f(|k|)$.

The full Spin-Boson Hamiltonian is then defined as
\begin{align}
\label{eq:H}
H:= H_0 + g V
\end{align}
for some  coupling constant $g>0$ on the
Hilbert space
\begin{align}
\label{hilbertspace}
\mathcal H := \mathcal K \otimes \mathcal F\left[ \mathfrak{h}\right] , \qquad
\mathcal K:= \C^2, 
\end{align}
where 
\begin{align}
\label{fockspace}
\mathcal F\left[ \mathfrak{h}\right] :=  \bigoplus^\infty_{n=0} \mathcal F_n\left[ \mathfrak{h}\right] 
,\qquad
\mathcal F_n\left[ \mathfrak{h}\right] := 
\mathfrak{h}^{\odot n},\qquad 
\mathfrak h:= L^2(\mathbb R^3,\C)
\end{align}
denotes the standard bosonic Fock space, and superscript $\odot n$ denotes the
n-th symmetric tensor product,  where by convention $\mathfrak{h}^{\odot 0}\equiv
\C$. Note that we identify $K\equiv K\otimes 1_{\mathcal F[\mathfrak h]}$ and
$H_f\equiv 1_{\mathcal K}\otimes H_f$ in our notation (see Remark \ref{R} below).  

Due to the direct sum, an element $\Psi \in
\mathcal F[\mathfrak{h}]$
can be represented as a family $(\psi^{ (n)})_{n\in\N_0}$ of wave functions $\psi^{ (n)} \in \mathfrak{h}^{\odot n}$. The state $\Psi$ with $\psi^{ (0)}=1$ and $\psi^{ (n)}=0$ for all $n\geq 1$ is called the vacuum and is denoted by
\begin{align}
\label{Omega}
\Omega:=(1,0,0,\dots)\in \mathcal F\left[ \mathfrak{h}\right] .
\end{align}
We define 
\begin{align}
\label{fock0}
\mathcal F_0:= \Big \{ \Psi=(\psi^{ (n)})_{n\in\N_0}\in
        \mathcal F[\mathfrak h] \,\big|\, 
        \exists  N  \in \N_0   :  
\psi^{ (n)}=  0 \, \forall n\geq N,   \forall n \in \N :
   \psi^{ (n)} \in \mathit S(\R^{3n }, \C)    \Big \},
\end{align}
where $\mathit S(\R^{3n }, \C) $ denotes the Schwartz space of  infinitely differentiable  functions with rapid decay. 

Then, for any $h\in \mathfrak{h}$, we define the operator $a(h): \mathcal F_0 \to \mathcal F_0$ by 
\begin{align}
\label{def:a}
 \left(a(h) \Psi  \right)^{ (n)}(k_1,...,k_n)&=\sqrt{n+1}
   \int \mathrm{d}^3k \, \overline{h(k)}\psi^{ (n+1)}(k,k_1,...,k_n) 
\end{align}
and $a(h)\Omega=0$. The operator $a(h)$ is closable and,  using  a slight abuse of notation, we denote  its   closure by the same symbol $a(h)$ in the following. The operator $a(h)$ is called the annihilation operator. The creation operator is defined as the adjoint of $a(h)$ and we denote it by $a(h)^*$. For $\Psi=(\psi^{ (n)})_{n\in\N_0} \in \mathcal F_0$, we find  that 
\begin{align}
\label{def:a1}
 \left(a(h)^* \Psi  \right)^{ (n)}(k_1,...,k_n)&=
 \frac{1}{\sqrt{n}}\sum^n_{i=1}  h(k_i)
   \psi^{ (n-1)}(k_1,...,\tilde k_i,...,k_n) ,
\end{align}
where the notation  $\tilde \cdot $ means that the corresponding variable is omitted.

 Occasionally,
we shall also use the physics notation and define the point-wise  creation and annihilation operators. The action of the latter in the $n$ boson sector  is to be
understood as:
\begin{align}
    \label{eq:aformal}
    \left(a(k) \Psi  \right)^{ (n)}(k_1,...,k_n)&=\sqrt{n+1}
    \psi^{ (n+1)}(k,k_1,...,k_n) , 
\end{align}    
 for $\Psi=(\psi^{ (n)})_{n\in\N_0}  \in \mathcal F_0$.
The operator $a(k)$ is not closable.
The point-wise creation operator $a(k)^*$ is only defined as a quadratic form on $\mathcal F_0$ in the following sense:
\begin{align}
\label{eq:akstar}
\left\langle \Phi, a(k)^* \Psi \right\rangle = \left\langle a(k)\Phi,  \Psi \right\rangle , \qquad \forall \Phi, \Psi \in \mathcal F_0.
\end{align}
Moreover, we define quadratic forms:
\begin{align}
\mathcal F_0 \times \mathcal F_0 \to \C, \quad (\Phi, \Psi) \mapsto \int \mathrm{d}^3k\, \overline{h(k)} \left\langle \Phi, a(k)\Psi \right\rangle
\end{align}
and 
\begin{align}
\mathcal F_0 \times \mathcal F_0 \to \C, \quad (\Phi, \Psi) \mapsto \int \mathrm{d}^3k\, h(k) \left\langle  \Phi, a(k)^* \Psi \right\rangle .
\end{align}
It is not difficult to see that these quantities are equal to $\left\langle\Phi, a(h)\Psi \right\rangle$ and $\left\langle\Phi, a(h)^*\Psi \right\rangle$, respectively.
The point-wise creation operator $a(k)^*$ is not defined as an operator  but, formally, we can express it in the following way:
\begin{align}
    \left(a(k)^* \Psi  \right)^{ (n)} (k_1,...,k_n)&=\frac{1}{\sqrt{n}}\sum^n_{i=1}
    \delta^{(3)}(k-k_i) \psi^{ (n-1)}(k_1,...,\tilde k_i,...,k_n)  . 
\end{align}
This is the usual formula that physicists use.
Here,   $\delta$ denotes the Dirac's delta tempered distribution acting on the Schwartz space of
test functions. 
Note that $a$ and $a^*$ fulfill the canonical commutation
relations:
\begin{align}
\label{eq:ccr}
  \forall h,l\in\mathfrak{h}, \qquad \left[a(h),a^*(l)   \right]=\left\langle h, l\right\rangle_2 , \qquad \left[a (h),a(l)   \right]=0 , \qquad \left[a^*(h),a^*(l)   \right]=0.
\end{align}
Let us recall some well-known facts about the introduced model. 
Clearly, $K$ is self-adjoint on $\mathcal K$ and its spectrum 
consists of two eigenvalues $e_0$ and $e_1$. The corresponding eigenvectors are
\begin{align}
\label{varphi}
\varphi_0= \left( 0,1 \right)^T \qquad \text{and} \qquad  \varphi_1= \left(1,0\right)^T \qquad \text{with} \qquad K \varphi_i =e_i \varphi_i , \quad i=0,1.
\end{align}
Moreover, $H_f$ is self-adjoint on its natural domain $\mathcal D(H_f)\subset \mathcal F[\mathfrak{h}]$ and its spectrum  $\sigma (H_f)= [0, \infty )$
 is absolutely continuous (see \cite{reedsimon2}). Consequently, the spectrum of $H_0$ is given by
$\sigma (H_0)=  [e_0, \infty )$, and $e_0,e_1$ are eigenvalues embedded in the absolutely continuous part of the spectrum of $H_0$ (see \cite{reedsimon1}).

Finally, also the self-adjointness of the full Hamiltonian $H$ is well-known
(see,  e.g.,  \cite{spohnspin}) and it can be shown using the standard estimate in
Lemma \ref{lemma:standardest1}. 
\begin{proposition}
\label{thm:Hsa}
The operator $gV$ is relatively bounded by $H_0$ with infinitesimal bound and, 
consequently, $H$ is self-adjoint and bounded below on the domain 
\begin{align}
\mathcal D(H) = D(H_0)=\mathcal K \otimes \mathcal D(H_f ),
\end{align}
i.e., there is a constant $b\in\R$ such that
\begin{align}
\label{bbdbelow}
b\leq H .
\end{align}
\end{proposition}
\begin{remark}\label{R} 
    In  this  work we omit
    spelling out   identity operators  whenever unambiguous. 
    For every vector spaces $V_1$,  $V_2$ and
    operators $ A_1 $ and $A_2$ defined on $V_1$ and $V_2$, respectively, we
    identify \begin{equation}\label{iden} A_1 \equiv A_1 \otimes \mathbbm
        1_{V_2}, \hspace{2cm}  A_2  \equiv \mathbbm 1_{V_1} \otimes A_2 .
    \end{equation}
    In order to simplify our notation further, and whenever
    unambiguous, we do not utilize specific  notations for every inner product
    or norm that we  employ.    
\end{remark}

\subsection{Access to the resonance: Complex dilation}
\label{sec:dil}
It is known (e.g., \cite{spohnspin}) that the only eigenvalue in the spectrum
of $H$ is 
\begin{align}
\label{def:gs}
\lambda_0:=\inf \sigma(H)
\end{align}
while the rest of the spectrum is absolutely continuous.  This implies that
there is no stable excited state in the massless Spin-Boson model.
Heuristically, the reason for this is that the atomic energy of the excited
state $e_1$ turns into what can be seen as a complex ``energy'' $\lambda_1$
with strictly negative imaginary part once the interaction is switched on (see
e.g.\ \cite{bach,bbp}). This complex energy $\lambda_1$ is referred to as 
resonance energy  and its imaginary part is responsible for the decay of the
excited state (see e.g.\ \cite{sigal,hasler}).    

Note that the ground state
$\Psi_{\lambda_0}$ of $H$ corresponding to ground state energy $\lambda_0$, i.e.,
\begin{align}
\label{gsprop}
H\Psi_{\lambda_0}=\lambda_0\Psi_{\lambda_0} ,
\end{align}
 has
already been constructed, e.g., in \cite[Theorem 1]{spohnspin}, \cite[Theorem 
1]{hasler1} and \cite[Theorem 3.5]{bbkm}.
Since $H$ on $\mathcal H$ is a
self-adjoint operator, $\lambda_1$ should rather be thought of as a complex
eigenvalue of $H$ on a bigger space than $\mathcal H$.  This prevents us from
being able to calculate the resonance energy directly by regular perturbation
theory on $\mathcal H$. The standard way to nevertheless get access to such a
resonance without leaving the underlying Hilbert space is the method of complex
dilation which will be introduced next.  We start by defining a family of
unitary operators on $\mathcal H$ indexed by $\theta \in\R$.
\begin{definition} 
\label{complextransf}
For $\theta \in \mathbb R$, we define the unitary
    transformation
\begin{align}
u_\theta: \mathfrak{h}&\to \mathfrak{h}
, \qquad \psi(k) \mapsto e^{-\frac{3\theta}{2}} \psi(e^{-\theta}k) .
\end{align}
Similarly, we define its canonical lift $U_\theta: \mathcal F [\mathfrak{h}]\to
\mathcal F [\mathfrak{h}]$ by the lift condition $U_\theta a(h)^* U_\theta^{-1}=a(u_\theta h)^*$, $h\in\mathfrak{h}$,   and $U_\theta \Omega=\Omega$.  This defines $U_\theta$
uniquely up to a phase which we  choose to equal one. With slight abuse of
notation, we also denote $\mathbbm 1_{\mathcal K}\otimes U_\theta$ on $\mathcal
H$ by the same symbol $U_\theta$. 
\end{definition}
 Thereby, we define the family of transformed Hamiltonians, for $\theta \in
 \R$,
\begin{align}\label{Hthetaaaa}
H^\theta :=U_\theta H U_\theta^{-1} =K + H^\theta_f +g V^\theta,
\end{align}
where
\begin{align}
    \label{eq:Hftheta}
H_f^\theta:= \int \mathrm{d^3}k \, \omega^\theta(k) a^*(k) a(k) , \qquad
V^\theta:= \sigma_1 \otimes  \left(a(f^{\overline \theta})+
a(f^{\theta})^*  \right)
\end{align} 
and 
\begin{align}
\label{def:thetafncts}
\omega^\theta(k):= e^{-\theta}|k|, \qquad f^\theta: \R^3\setminus \{0\}\to\R , \quad k\mapsto e^{-\theta (1+\mu)} e^{-e^{2\theta}\frac{k^2}{\Lambda^2}}|k|^{-\frac{1}{2}+\mu}.
\end{align}
  Eqs.\ \eqref{def:thetafncts}, \eqref{eq:Hftheta} and the
right hand side of Eq.\ \eqref{Hthetaaaa} can be defined for complex 
$\theta$. 
If  $ |\theta| $ is small enough,  $ K + H^\theta_f +g V^\theta $ is a closed
(non self-adjoint) operator. However, the middle term in Eq.\ \eqref{Hthetaaaa}
is not necessarily correct because although $ U_{\theta} $ can be defined for
complex  $\theta$, it turns out to be an unbounded operator, and $  U_\theta H
U_\theta^{-1} $ might not be densely defined.

We say that 
$ \Psi $ is an analytic vector  if the map $ \theta
 \mapsto \Psi^\theta := U_\theta \Psi $  has an analytic continuation  from an open connected set in the real line to a (connected) domain in the complex plane. In general we will not specify their domains of analyticity (it will be clear from the context). It is well-known that there is a dense set of entire vectors (they are analytic in $\mathbb{C}$).   This result has been proven in a variety of similar models, for example, in
\cite{bach,jaksic}. For the sake of completeness, we provide a proof in
Appendix \ref{app:dil}. Furthermore, we define the open disc 
\begin{align}
\label{eq:def-disc}
D(x, r):=\left\{ z\in\C : |z-x|<r  \right\}   \qquad  x\in \C , r>0 ,
\end{align}
and note that for $\theta \in  D(0, \pi/16)$ we have
\begin{align}
\norm{ V^\theta \left( H_0 +1  \right)^{-\frac{1}{2}}}&\leq 
\norm{f^\theta}_2 +2\norm{f^\theta/\sqrt{\omega}}_2
\label{eq:Vtheta-est}
\end{align}
which is guaranteed by the standard estimate \eqref{eq:Vest} given in
Appendix~\ref{app:sa}, since \eqref{def:thetafncts} together with the special
choice $\theta \in  D(0, \pi/16)$  imply that  $f^\theta , f^\theta/\sqrt{\omega}
\in \mathfrak{h}$. Hence, for $\theta\in D(0,\pi/16)$ 
the operators $H^\theta$ are densely defined and closed. Moreover, the analytic properties
of this family of operators
in $g$ and $\theta$ are known:
\begin{lemma}
\label{typea}
The family $\left\lbrace  H^\theta \right\rbrace_{\theta\in \R}$ of unitary
equivalent, self-adjoint operators with $\mathcal D(H^\theta)=\mathcal D(H)$
extends to an analytic family of type A for $\theta\in D(0, \pi/16)$. 
\end{lemma}
The above result was proven for the Pauli-Fierz model in  \cite[Theorem
4.4]{bach}, and with small effort that proof can be adapted to our
setting. 
\begin{lemma}
    \label{lemma:specdilh0} 
    Let $\theta\in \C$.  Then, 
        $\sigma (H^\theta_0)= 
 \left\lbrace e_i + e^{-\theta} r : r\geq 0 , i=0,1 \right\rbrace$.  
\end{lemma}
We provide a proof in Appendix \ref{app:dil}. 
For sufficiently small coupling constants and  
for $\theta \in \mathcal S$, where $\mathcal S$ is   the   subset of the complex plane defined in    Eq.\  \eqref{def:setS}      below,
it has been shown that $H^\theta$ has two non-degenerate eigenvalues
$\lambda^\theta_0$ and $\lambda^\theta_1$ with corresponding rank one
projectors denoted by $P^\theta_0$ and $P^\theta_1$, respectively; see, e.g.,
\cite[Proposition 2.1]{bdh-res}.
Note that there the $\theta$-dependence was omitted in the notation. For convenience of the reader, we make it explicit in this paper.
 The corresponding dilated eigenstates can,
therefore, be written as
\begin{align}
\label{eq:gsvec}
\Psi^\theta_{\lambda_i}:=  P_i^\theta    \varphi_i\otimes \Omega  , \qquad i=0,1 ,
\end{align}
where the eigenstates $\varphi_i$ of the free atomic system    are  given in
\eqref{varphi}, and $\Omega$ is the bosonic vacuum defined in \eqref{Omega}.
In our notation $\Psi^\theta_{\lambda_i}$ is not necessarily normalized.  We
know from \cite[Theorem 2.3]{bdh-res} that the eigenvalues $\lambda^\theta_i$
are independent of $\theta$
 as long as  $\theta $ belongs to the set $\mathcal S$, 
 and therefore, we suppress it in our notation
writing $\lambda^\theta_i\equiv\lambda_i$. Note that this is not true for the
eigenstates $\Psi^\theta_{\lambda_i}$.   In \cite{bdh-res} (as well as in   Eq.\   \eqref{def:setS}    below)
we choose an open connected set $\mathcal{S}$ that does not include $0$ (the imaginary parts of the points in this set are bounded from below by a fixed positive constant). We chose such a set in order to have a single set $\mathcal{S}$  for the cases $i = 0 $  and $i =1$, because we want to keep our notation as simple as possible (otherwise a two cases formulation would propagate all over our papers).  However, the fact that $0$ is not contained in $ \mathcal{S} $ is only necessary for the case $i = 1$ (the resonance -   due to the self-adjointness of $H$ the
state $\Psi^\theta_{\lambda_1}$ can not even exist  for
$\theta=0$).  For the case $i= 0$ (the ground state) we can choose instead a connected open set containing $0$. In this set, it is still valid that $\lambda^\theta_0$ does not depend on $\theta$ and, therefore, it equals the ground state energy, and  $\Psi^{\theta=0}_{\lambda_0}= \Psi_{\lambda_0}$ - as introduced above.  This is explained in    \cite[Remark 2.4]{bdh-res}.

\subsection{Scattering theory}
\label{sec:scattering}
Finally, we give a short review of scattering theory which will be necessary to
state the main results  in
Section~\ref{sec:mainresult}.

The first obstacle in formulating a scattering theory of a second-quantized
system lies in the definition of the wave operators. Unlike in first-quantized
quantum theory, where one defines the scattering operator to be $S:=\Omega^*_+
\Omega_-$ with the wave operators $ \Omega_{\pm} $ given by the strong limits
$\Omega_{\pm}:=s\mbox{-}\lim\limits_{t\rightarrow \pm \infty}
e^{itH}e^{-itH_0}$, in quantum field theory, the corresponding wave operators do not exist in a
straight forward sense.  Instead, one establishes the existence of the
asymptotic annihilation and creation operators first, which can then be used to
define the wave operators.

\begin{definition}[Basic components of scattering theory]
\label{defasymptop}
We denote the dense subspace of compactly supported, smooth, and complex-valued
functions on $\mathbb
R^3\setminus \{0 \}$ in $\mathfrak h$  by
\begin{align}
    \label{def:h0}
    \mathfrak{h}_0 :=\mathit C_c^\infty (\mathbb R^3\setminus \{0 \},\C) \subset
    \mathfrak{h}.
\end{align}
Furthermore, we define the following objects:
\begin{enumerate}
    \item[(i)]  
    For $h\in\mathfrak{h}_0$ and $\Psi \in \mathcal K\otimes \mathcal D(H_f^{1/2})$, the asymptotic annihilation operators
\begin{align}
    \label{asymptop}
    a_\pm(h)\Psi :=\lim\limits_{t\to\pm \infty}a_t(h)\Psi, \quad
    a_t(h):=e^{itH}a(h_t) e^{-itH},
    \quad
  h_t(k):=h(k) e^{ - it\omega(k)} .
\end{align}
The existence of this limit is proven in Lemma \ref{lemmadiff} (i) below.
Moreover, we define the asymptotic creation operators
 $a_\pm^*(h)$ as the respective adjoints.
\item[(ii)] The asymptotic Hilbert spaces 
\begin{align}
\label{asympthilbert}
\mathcal H^\pm :=\mathcal K^\pm \otimes \mathcal F\left[\mathfrak{h}\right]
\quad \text{where} \quad  \mathcal K^\pm:=\left\lbrace \Psi\in \mathcal H :
a_\pm(h) \Psi=0 \,\,\,\, \forall h\in \mathfrak{h}_0    \right\rbrace .
\end{align} 
\item[(iii)] The wave operators 
\begin{align}
\label{intertwining}
&\Omega_\pm :\mathcal H^\pm \to \mathcal H
\\ \notag
&\Omega_\pm \Psi \otimes a^*(h_1)...a^*(h_n) \Omega:=a^*_\pm (h_1)...a^*_\pm (h_n) \Psi, \quad h_1,...,h_n \in \mathfrak{h}_0,  \quad \Psi \in \mathcal K^\pm .
\end{align}
\item[(iv)] The scattering operator  
 $S:=\Omega^*_+\Omega_-$.
\end{enumerate}
\end{definition}
The limit operators $a_\pm$ and $a_\pm^*$ are called asymptotic
outgoing/ingoing annihilation and creation operators. The existence of the
limits in \eqref{asymptop},
their properties, especially that $\Psi_{\lambda_0}\in\mathcal K^\pm$ and
$\Omega_\pm$ are well-defined, are well-known facts (see e.g.
\cite{fau1,fau2,fau3,fgs1,fgs2,rgk,rk,rk2,derezinski,bkz}). For the convenience
of the reader,  Lemma \ref{lemmadiff}  collects all relevant facts and we provide simplified proofs
for our setting in Appendix~\ref{app:welldef}. 
We can thus define the following two-body scattering matrix coefficients:
\begin{align}
\label{eq:2bodyscat}
S(h,l)= \norm{\Psi_{\lambda_0}}^{-2}\left\langle
a^*_+(h)\Psi_{\lambda_0},a^*_-(l) \Psi_{\lambda_0} \right\rangle, \qquad \forall
h,l\in \mathfrak h_0 ,
\end{align}
where the factor $\norm{\Psi_{\lambda_0}}^{-2}$ appears due to the fact that,
as already mentioned above,  in our notation, the ground state
$\Psi_{\lambda_0}$ is not necessarily normalized.  In addition, it will be
convenient to work with the corresponding two-body transition matrix
coefficients given by
\begin{align}
T(h,l)=S(h,l) -\left\langle h , l \right\rangle_2 \qquad \forall h,l\in \mathfrak h_0 .
\label{eq:Tmatrix}
\end{align}
These matrix coefficients carry a ready physical interpretation as transition
amplitudes of the scattering process in which an incoming boson with wave
function $l$ is scattered at the two-level atom into an outgoing boson with
wave function $h$.  Notice that the transition matrix coefficients of
multi-photon processes can be defined likewise but in this work we focus on
one-photon processes only.

 It has been shown in \cite{spohnspin} that
the spectrum of $H$   contains only  one eigenvalue $\lambda_0$ (and it is non-degenerate), 
namely the ground state energy, and the  rest of the spectrum of $H$ is
absolutely continuous.  
In case  that   asymptotic completeness holds, i.e.
\begin{align}
\mathcal K^\pm = \text{Ran} \left(\chi_{\text{pp}}(H)\right),
\end{align} 
 all one-boson processes are of the
form \eqref{eq:2bodyscat}. Here, 
  $\text{Ran} \left(\chi_{\text{pp}}(H)\right)$ denotes  the  states associated
with pure points in the spectrum of $H$.

Asymptotic completeness has actually been proven in
\cite{rgk,rk,rk2} for the Hamiltonian $H$ defined in \eqref{eq:H}, however,
with coupling functions $f\in \mathit C^3_c(\R^3\setminus\{0\},\C)$, i.e., the  functions 
that are three times continuously differentiable and have compact support.
 In our case, we need an analytic continuation of our Hamiltonian in order to study resonances. This implies that the coupling function  $f$ cannot be  compactly supported (see  \eqref{eq:f}), 
 however it belongs to the  Schwartz space.  We expect asymptotic completeness also to hold in our
case, although our results do not depend  on it.  

\section{Main result}
\label{sec:mainresult}

We are now able to state our main results. The corresponding proofs will be
provided in Section \ref{sec:proof-mainresult} after we review a list of
necessary results of a previous work \cite{bdh-res} in
Section~\ref{sec:strategyproof}.

First, we state a definition that we use for our main result
\begin{definition}
\label{def:G}
 Using solid angles $\mathrm d\Sigma, \mathrm d\Sigma'$, we define, for all $h,l\in\mathfrak{h}_0$,
        \begin{align}
            \label{eq:G-def}
            G: \R \to \C , \qquad r \mapsto G(r):=
            \begin{cases}
                \int \mathrm{d}\Sigma \mathrm{d}\Sigma' \,  r^4  \overline{h(r,\Sigma)} l(r,\Sigma') f(r)^2  \qquad &\text{for} \quad r\geq 0
                \\
                0 \quad &\text{for} \quad r<0 .
            \end{cases} 
        \end{align}
\end{definition}  
We recall the definition $E_1  =g^{-2}  \Im \lambda_1 $ given in  \eqref{def:e1}.   It follows from Eqs.\ \eqref{EI} and \eqref{eq:impartres1}  below that $E_1=E_I +g^{ a  }\Delta $ where  $a>0$,  $\Delta\equiv \Delta(g)$ is  uniformly bounded and $E_I<0$ is the constant defined in \eqref{EI}. 
This implies that 
\begin{equation}\label{Im}
E_1 \leq  -  \boldsymbol{c} < 0,
\end{equation}
for some constant $ \boldsymbol{c} $ that does not depend on $g$ (for small enough $g$).

Our  main result provides
a relation between the scattering matrix element and the complex
dilated resolvent of the Hamiltonian.
\begin{theorem}[Scattering formula]
\label{FKcor}
 There is a constant $\boldsymbol g>0$ such that for every $g\in (0,\boldsymbol g]$,   $ \theta $ in the set
    $\mathcal{S} $ defined  in \eqref{def:setS} below,
and
 for all $h,l\in\mathfrak{h}_0$, the
two-body transition matrix coefficients are given by
 \begin{align} 
\label{scatteringformulapp}
T(h,l)= & T_{P}(h,l) + R(h,l) ,
\end{align} 
where 
\begin{align} 
\label{scatteringformulapp1}
T_{P}(h,l):= &    4 \pi i g^2  \norm{\Psi_{\lambda_0}}^{-2}  \int \mathrm{d}r     \,G(r)
\left(  \frac{ \Re \lambda_1 -  \lambda_0 }{ \big (  r + \lambda_0 - \lambda_1 \big ) \big (  r -  \lambda_0 + \overline{\lambda_1}    \big )  }   \right)  
\\  =  &  \notag   M  \int \mathrm{d}r     \,G(r)
\Big (  \frac{ E_1 g^2 }{ \big (  r + \lambda_0   
 - \Re \lambda_1
- i g^2 E_1  
 \big ) \big (  r -  \lambda_0 + \overline{\lambda_1}    \big )  }  \Big )  , 
\end{align}
  and there is a constant  $C(h,l)$ (that does not depend on $g$) such that 
\begin{align}  
\label{scatteringkernel}
|R(h,l) |\leq C(h,l)  g^3 |\log g|  .
\end{align}
Here, we use the notation
\begin{align}
\label{constM}
 M   :=     4 \pi i (  \Re \lambda_1 - \lambda_0 ) E_1^{-1 }  \norm{\Psi_{\lambda_0}}^{-2}.
\end{align}
 $  T_P(h,l) $ is the leading term   in terms of powers of $g$ for small $g$, and $R(h,l)$ is regarded as the error term.   This is   justified by Remark \ref{rem:order} below.

Our proof permits us to find an explicit formula for the dependence of $C(h,l)$ on $ h $ and $l$, see Remark \ref{Constant} below.
\end{theorem} 
\begin{remark}
\label{rem:order}
 The scattering processes described  by the transition matrix in \eqref{scatteringformulapp} clearly depend on the incoming and outgoing photon states, $ l$ and $h$. This is well understood from the  physics as well as the mathematics perspectives. For example, it can be read from  \eqref{def:G} that if $ l  $ is supported in a ball of radius $t$ and $h$ is supported in its complement, then the principal term $T_{P}(h, l )$ vanishes and only higher order terms (with respect to powers of $g$) contribute to the scattering process. The quantity $T_{P}(h, l )$ is the only one that might produce scattering processes of order $g^2 $ since the remainder is of order  $ g^3 |\log g|   $.   If an experiment is  appropriately prepared, then such an scattering process will be observed and the term describing this is $T_P(h,l)$.  This justifies why we call it the leading order (or principal) term. In Appendix  \ref{app:order}  give an example of a large class of functions $h$ and $l$ that make $ T_P(h,l ) $ larger or equal than a strictly positive constant times $g^2$.  In particular, we prove that this happens when the corresponding function  $G$ is positive and strictly positive at 
$ \Re \lambda_1 - \lambda_0  $.   
\end{remark}
\begin{remark}
By Eqs.\ \eqref{scatteringformulapp1} and \eqref{eq:G-def}, we can express 
 the principal term $T_P(h,l)$ in terms of an integral kernel:
\begin{align}
\label{scat_intker1}
T_P(h,l)=\int\mathrm{d}^3k \mathrm{d}^3k' \, \overline{h(k)}l(k')\delta(|k|-|k'|)T_P(k,k') ,
\end{align} 
where 
\begin{align}
\label{scat_intker}
T_P(k,k')=   M    f(k)f(k')  \left(  \frac{ E_1 g^2 }{ \big (  |k'| + \lambda_0   
 - \Re \lambda_1
- i g^2 E_1  
 \big ) \big (  |k'| -  \lambda_0 + \overline{\lambda_1}    \big )  }  \right) .
\end{align}
Eq. \eqref{scat_intker1} is important, because it allows us to calculate the leading order of the scattering cross section. It is proportional to the modulus squared of
$   T_P(k,k') $:
\begin{equation}
\label{Tsquared}
|T_P(k,k')|^2 =\left( \frac{|  M   |^2 |   f(k) |^2 |f(k')  |^2}{ |   |k'| -  \lambda_0 + \overline{\lambda_1}   |^2  } \right)    \frac{ E_1^2 g^4 }{ ( |k'| + \lambda_0   
 - \Re \lambda_1 )^2
+ g^4 E_1^2  } .   
\end{equation}
 For momenta $|k'| $ in a neighborhood of  $\Re \lambda_1 - \lambda_0$, the behavior in the expression  above is dominated by the Lorentzian function. As expected,  there is a maximum  when the energy of  the incoming photons is close to the difference of the resonance and the ground state energies of the system and the width of this peak is controlled by the imaginary part of the resonance $\Im\lambda_1$.  
\end{remark}

Note that the Dirac's delta distribution in \eqref{scat_intker1} is to be understood as the expression in \eqref{scatteringformulapp1}. Notice that \eqref{scat_intker} is not defined
for $k=0$ or $k'=0$. However, since we take $h,l\in \mathcal
C^\infty_c(\R^3\setminus \{0\},\C)$, the expression \eqref{scat_intker1}
is well-defined.  Similar
distribution kernels in a  related model have been studied in
\cite{bkz,bfp}. 
\begin{remark}\label{porquee}
    In this work we denote by $ C $ any
generic (indeterminate) constant that might change from line to line. This
constants do not depend on the coupling constant and the auxiliary parameter
   $ n$ introduced in
Section \ref{Saux}.
\end{remark} 

\section{Known results on spectral properties and resolvent estimates} 
\label{sec:strategyproof}

In this section we present  results about the spectrum of the dilated
Spin-Boson Hamiltonian and resolvent estimates proven   in our previous paper
\cite{bdh-res}. Here, we do not repeat proofs but  give
precise references     for them.   We collect only properties and estimates
     for  the model  under consideration that are necessary for
    the proofs of our main theorems.   

 Throughout   this paper we address   the case of small coupling, i.e., we 
    assume  the coupling constant $g$ to be sufficiently
    small.  The restrictions on the coupling constant only stem from the
    requirements needed to prove the results reviewed in this section, i.e.,
    the ones considered in \cite{bdh-res}. We do not  explicitly 
    specify how small the coupling constant must be 
    but  give  precise  references from which  such bounds  can be inferred. This  issue  is addressed
  by the next definition:

\begin{definition}[Coupling Constant]\label{gggg}
      Throughout this work  we  assume that  $ g
    \leq  \boldsymbol{g} $, where $0<  
    \boldsymbol{g} $  
      satisfies    Definition   4.3  and Eq.\ (5.58) in \cite{bdh-res}, the Fermi-golden rule
    (see Eqs.\ \eqref{eq:impartres1} and \eqref{eq:impartres} below) and Eq.\
    \eqref{ground1} below.      
\end{definition}  
 We denote the imaginary part of the dilation parameter $\theta$ by 
\begin{align}
\label{eq:nu}
\nu : = \Im \theta
\end{align}
and  assume that $\theta $ belongs to the set 
\begin{align}
\label{def:setS}
 \mathcal S:=\left\{\theta\in\C: -10^{-3} <   \Re \theta < 10^{-3} \text{ and }
\boldsymbol{\nu} < \Im \theta  < \pi/16 \right\} ,
\end{align}
 where $\boldsymbol \nu \in (0, \pi/16)$ is a fixed number (see \cite[Definition 1.4]{bdh-res}). 

\subsection{Spectral estimates}\label{Sspec}

We know from \cite[Proposition 2.1]{bdh-res} that the Hamiltonian $H^{\theta}$  has two eigenvalues $\lambda_0 $ and $ \lambda_1 $ in small neighborhoods of $ e_0 $ and $e_1$, respectively.   Loosely put, $ e_0 $ turns into the ground state $ \lambda_0$ and $ e_1 $ tuns into the resonance $ \lambda_1$ once  the interaction is tuned on. Both $\lambda_0$ and $  \lambda_1$ do not depend on $\theta$ provided that $ \theta \in \mathcal S$ and in the case of  $ \lambda_0  $ we can take $ \theta  $ in a neighborhood of $0$ and, therefore, infer that $ \lambda_0 $ is real and gives the ground state  energy. This is  proven   in    \cite[Theorem 2.3]{bdh-res} and \cite[Remark 2.4]{bdh-res}.   

In \cite[Theorem 2.7]{bdh-res},   we give a very  sharp 
estimation of the  location of the  spectrum of 
$ H^\theta $. We prove, among other things  that,
locally, in neighborhoods of $ \lambda_0 $ and $
\lambda_1$, its spectrum is contained in cones with
vertices at   $ \lambda_0 $ and $ \lambda_1$.  To make this statement more precise we  need to  introduce some
 more  concepts and notation.  There are two auxiliary parameters that play  an  important role  in our constructions:
\begin{align}\label{rhos}
\rho_0 \in (0,1), \qquad  \rho \in  (0,\min e_1 / 4  ), 
\end{align}
 which also satisfy the conditions in
\eqref{dorm2} below. In order to specify the spectral properties of $
H^{\theta} $ we  define some regions in the complex plane:  
\begin{definition}
\label{def:regionsAB}
For fixed $\theta\in \mathcal S$, we set $\delta = e_1- e_0 = e_1$ and define the regions 
\begin{align}
\label{region:A}
A:&=
A_1\cup A_2\cup A_3  ,
\end{align}
where
\begin{align}
A_1:&=\left\{  z\in\C : \Re z <e_0-\delta/2 \right\}
\\
A_2:&= \left\{  z\in\C : \Im z >\frac{1}{8}\delta \sin (\nu) \right\}
\\
A_3:&= \left\{  z\in\C : \Re z >e_1+\delta/2 , \Im z \geq -\sin (\nu/2) \left(\Re (z) -(e_1+\delta/2)   \right)\right\} ,
\end{align}
and for $i=0,1$, we define
\begin{align}
\label{region:Bi1}
B_i^{(1)}:=\left\{  z\in\C : |\Re z-e_i| \leq \frac{1}{2}\delta, -\frac{1}{2}
\rho_1 \sin(\nu)\leq \Im z \leq \frac{1}{8}\delta \sin (\nu)  \right\} .
\end{align}
These regions are depicted in Figure \ref{fig:regionsAB}.
\end{definition}
\begin{figure}[h]
\centering
\includegraphics[width=\textwidth]{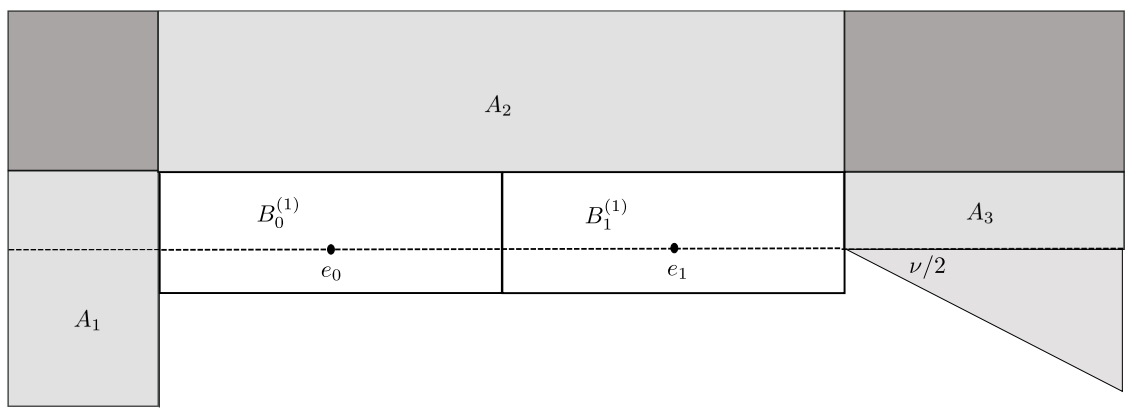}
  \caption{An illustration of the subsets of the complex plane introduced in
  Definition \ref{def:regionsAB}.}
    \label{fig:regionsAB}       
\end{figure}
 For a fixed  $ m \in \mathbb{N}, \:  m \geq 4,$ we define the cone
\begin{align}
\label{eq:defcone}
\mathcal C_m(z) :=\left\{  z+xe^{-i\alpha} : x\geq 0 ,
|\alpha-\nu |\leq  \nu/m \right\} .
\end{align} 
It follows from the induction scheme in \cite[Section 4]{bdh-res} that $ \lambda_i \in  B_i^{(1)} $, and moreover,
 \cite[Theorem 2.7]{bdh-res}  together with \cite[Lemma 3.13]{bdh-res}  yields
\begin{align}\label{spectrum}
\sigma(H^{\theta}) \subset \mathbb{C} \setminus  \Big [  A \cup \big ( B_0^{(1)} \setminus  \mathcal C_m(\lambda_0)   \big ) \cup  \big ( B_1^{(1)} \setminus  \mathcal C_m(\lambda_1)   \big )      \Big ].
\end{align}
As we mention above, we have $ \lambda_0 \in \mathbb{R}$.
The imaginary part of $\lambda_1$ can be also estimated (see \cite[ Remark
2.2]{bdh-res}  --  Fermi golden rule):
 Recalling \eqref{eq:f}, we  define
 \begin{align}
 \label{EI}
 E_I:=-4\pi^2 (e_1-e_0)^2 |f(e_1-e_0)|^2 .
 \end{align}
  Then, for $g$ small enough, there are  constants $C,  a   >0$  such that
\begin{align}
\label{eq:impartres1}
\left|  \Im \lambda_1 -g^2 E_I \right| \leq g^{2+ a } C .
\end{align}
This implies that, for $g$ small enough, there is   constant  $\boldsymbol{c}>0$ such that
\begin{align}
\label{eq:impartres}
\Im \lambda_1 <-g^2  \boldsymbol{c} <0 .
\end{align}

\subsection{Auxiliary (infrared cut-off) Hamiltonians } \label{Saux}

 Some of the bounds in Section \ref{sec:proof-mainresult} employ a certain approximation of the Hamiltonian
$H^\theta$ by Hamiltonians with infrared cut-offs. The strategy will be the following: A   
mathematical expression that depends on  $ H^{\theta} $ is replaced by a
corresponding one that depends on  a particular  infrared cut-off
Hamiltonian. We  then  analyze the infrared cut-off expression and  estimate the difference between
both expressions. The construction of a sequence of infrared cut-off
Hamiltonians $(H^{(n), \theta})$  such that, as $n$ tends to infinity,  the
cut-off is removed is called multiscale analysis. In \cite{bdh-res}, we
present the full details of  this  method and derive several results. Here, we
only use some of  those  results and only present  the  notation necessary to review this part of \cite{bdh-res}. The infrared cut-off
 Hamiltonians  $ H^{(n), \theta} $  are parametrized 
by a sequence of numbers (see also  \eqref{rhos} and \eqref{dorm2})
\begin{align}\label{rho22}
\rho_n :=  \rho_0 \rho^{n},
\end{align}
 where  the Hamiltonians $ H^{(n), \theta} $   are  defined
by
\begin{align}
\label{Hntheta}
H^{(n),\theta} :&= K  +  H^{(n),\theta}_f +g V^{(n),\theta}=:H^{(n),\theta}_0+g V^{(n),\theta}
\\
\label{Hfntheta}
H_f^{(n),\theta}:&= \int_{\R^3\setminus \mathcal B_{\rho_n}} \mathrm{d^3}k \, \omega^\theta(k) a^*(k) a(k) , \quad \omega^\theta(k)= e^{-\theta}|k|
\\
\label{Vntheta}
V^{(n),\theta}:&=\sigma_1 \otimes \int_{\R^3\setminus \mathcal B_{\rho_n}}\mathrm{d^3}k \,\left(f^{\theta}(k)a(k)+ f^\theta(k) a^*(k) \right) ,
\\
f^\theta: & \,\R^3\setminus \{0\}\to\R , \quad k \mapsto e^{-\theta (1+\mu)} e^{-e^{2\theta}\frac{k^2}{\Lambda^2}}|k|^{-\frac{1}{2}+\mu} ,
\end{align}
on the Hilbert space
\begin{align}
\label{hilbertn}
\mathcal H^{(n)}:=\mathcal K \otimes \mathcal F[\mathfrak{h}^{(n)}] , \quad
\mathfrak{h}^{(n)}:=L^2(\R^3\setminus \mathcal B_{\rho_n},\C), \quad  \mathcal B_{\rho_n}:=\left\{  x\in\R^3  : |x|<\rho_n \right\} .
\end{align}
Additionally, we  define 
\begin{align}
\label{Htilde}
    \tilde H^{(n),\theta} := H^{\theta}_0 +gV^{(n),\theta}
\end{align}
 and  fix the Hilbert  spaces 
\begin{align}
\label{hilbertrest}
\mathfrak{h}^{(n,\infty)}:= L^2(\mathcal B_{\rho_n} ) \quad \text{and} \quad  \mathcal F[\mathfrak{h}^{(n,\infty)}], 
\end{align}
 defined as in \eqref{fockspace} with $  \mathfrak{h}^{(n,\infty)} $ instead of  $  \mathfrak{h} $, 
with vacuum states $ \Omega^{(n , \infty)} $ and corresponding orthogonal projections 
$ P_{ \Omega^{(n , \infty)}}$. 
Note that
$\mathcal H \equiv \mathcal H^{(n)}\otimes \mathcal
F[\mathfrak{h}^{(n,\infty)}]$.

In \cite[Proposition 2.1]{bdh-res} and \cite[Theorem 4.5]{bdh-res}, we prove  that, for each $n\in\N$, $H^{(n), \theta}$
has isolated eigenvalues $ \lambda_i^{(n)} $ in certain neighborhoods of $ e_i $, for
$ i \in \{0,1 \}$, respectively. The fact that
 these eigenvalues  are isolated permits us to define their
corresponding  Riesz projections which are denoted  by 
\begin{align}
\label{projections}
P_i^{(n)} \equiv P_i^{(n), \theta}. 
\end{align}
In  \cite[Proposition 2.1]{bdh-res}, we prove that this sequence of projections  converges   to the projection associated to the eigenvalue $\lambda_i$, i.e.,
\begin{align}\label{ana}
P_{i}^{\theta} \equiv  P_i =  \lim_{n \to \infty}P_i^{(n),\theta}  \otimes P_{\Omega^{(n, \infty)}}  , 
\end{align} 
and that the latter is analytic with respect to $\theta$ (see  \cite[Theorem 2.3]{bdh-res}).      
Furthermore, it follows from \cite[Remark 5.11]{bdh-res} that
\begin{align}\label{projj}
\Big \|  P^\theta_i -  P_i^{(n),\theta}  \otimes P_{\Omega^{(n, \infty)}} \Big \| \leq  2  \frac{g}{\rho} \rho_n^{\mu /2} \leq     \rho_n^{\mu /2}. 
\end{align}
This  together with \cite[Lemma 3.6]{bdh-res} implies that there is a constant $C$ such that
\begin{align}\label{projj1}
\Big \|  P^\theta_i -   P_{ \varphi_i } \otimes P_{\Omega} \Big \| \leq C g,
\end{align}
and in addition,  we know from  \cite[Lemma 4.7]{bdh-res}  that
\begin{align}\label{projj2}
\Big \|   P_i^{(n),\theta}   \Big \| \leq 3,
\end{align}
for every $n\in\N$.  Finally, \cite[Lemma 5.1]{bdh-res} yields that for all $n\in\N$
\begin{align}\label{ground}
| \lambda_i -  \lambda_i^{(n)} | \leq 2 g \rho_n^{1+ \mu /2}.
\end{align}
This together with \cite[Lemma 3.10]{bdh-res}, which states that there is a constant $C $ such that
$  | e_i -  \lambda_i^{(1)} |  <  C g$, proves   that there is a constant $  C  $ such that, for every $n \in \mathbb{N}$ and for g sufficiently small, we have
\begin{align}\label{ground1}
| \lambda^{(n)}_i - e_i | \leq C g \leq 10^{-3} e_1 , \qquad | \lambda_i - e_i | \leq C g \leq 10^{-3} e_1.
\end{align}

\subsection{Resolvent estimates}\label{Sress}

In  \cite{bdh-res}, we derive bounds for the resolvent of $  H^{\theta}$ in     
 $            \Big [  A \cup \big ( B_0^{(1)} \setminus  \mathcal C_m(\lambda_0)   \big ) \cup  \big ( B_1^{(1)} \setminus  \mathcal C_m(\lambda_1)   \big )      \Big ]   $, see  \eqref{spectrum}.  The region $A$ is far away from the spectrum, and therefore,  resolvent estimates in this  region are easy. In \cite[Theorem 3.2]{bdh-res},  we prove that there is a constant $C$ such that 
\begin{align}\label{resA}
\Big \| \frac{1}{H^{\theta} - z}  \Big \| \leq C \frac{1}{ |z  - e_1|}, \qquad
\forall z \in A. 
\end{align} 
Resolvent estimates in the regions $   B_0^{(1)} \setminus  \mathcal C_m(\lambda_0)      $ and  $   B_1^{(1)} \setminus  \mathcal C_m(\lambda_1)      $ are much more complicated because these regions share boundaries with the spectrum.

In  \cite[Theorem 5.5]{bdh-res}, we prove that, for $i\in\{0,1\}$,
  $B_i^{(1)} \setminus  \mathcal C_m \left(\lambda_i^{(n)} + (1/4)
\rho_n e^{-i\nu}\right) \setminus \{ \lambda^{(n)}_i  \}  $ is contained in the resolvent set of $H^{(n), \theta}$ 
and that there is a constant $ \boldsymbol C$ such that 
\begin{align}\label{dedo}
\norm{\frac{1}{ H^{(n), \theta}-z}  \overline{P_{i }^{(n),\theta}} }  \leq 
  \boldsymbol{   C}^{n+1}  \frac{1}{ {\rm dist} ( z, \mathcal{C}_m(\lambda_i^{(n)} +  
(1/4)\rho_n e^{-i\nu} )       )} ,
\end{align} 
for every $ z  \in  B_i^{(1)} \setminus  \mathcal C_m \left(\lambda_i^{(n)} +  
(1/4)\rho_n e^{-i\nu}\right) $, where $ \overline{P_{i }^{(n),\theta}} =  1 - P_{i }^{(n),\theta}   $. Here, the symbol ${\rm dist}$ denotes the Euclidean distance in $\C$. In  \cite{bdh-res}, we select   the auxiliary numbers $ \rho$
and $ \rho_0 $ satisfying  $   \boldsymbol{C}^8 \rho_0^{\mu} \leq 1, $  and  $ \boldsymbol{C}^4 \rho^{\mu} \leq 1/4        $. In this paper we assume the stronger    conditions 
\begin{align}\label{dorm2}
\boldsymbol{C}^8 \rho_0^{\mu} \leq 1, \qquad \boldsymbol{C}^8 \rho^{\mu} \leq 1/4 ,  \qquad   (\text{and hence} \qquad  \boldsymbol{C} \rho^{ \frac{1}{2}  \iota (1 + \mu/4)} \leq 1),    
\end{align}  
where  
\begin{align}\label{dorm222}
\iota = \frac{\mu/4 }{ ( 1+ \mu/4 ) } \in (0, 1).   
\end{align}  
The constant $ \boldsymbol{C} $ is larger that $ 10^6$, it is specified in
Definition 4.1 and Eq.\ (5.58) in \cite{bdh-res}, however, its precise form is not relevant in this paper (in \cite{bdh-res}, we do not intend to calculate optimal constants, because this would make the  work harder to read).  From the inequalities above and  Eq.\ \eqref{rhos} we obtain that, for very $n \in \mathbb{N}$: 
\begin{align}\label{rhoneuno}
\rho_n \leq 10^{-6} e_1. 
\end{align}
Finally, we prove in  \cite[Theorem 5.9]{bdh-res} that
the set  $  \in   B^{(1)}_i \setminus \mathcal C_m(\lambda_i^{(n)} -  e^{-i \nu} \rho_n^{1+ \mu/4} )  $
 is contained in the resolvent set  of both
 $   H^{ \theta} $ and $ \tilde H^{(n), \theta} $ 
and for all $z $ in this set  there is a constant $C$  such that: 
\begin{align} \label{diffet}
\norm{ \frac{1}{ H^{\theta}-z}     -   \frac{1}{ \tilde H^{(n), \theta}-z}   }
\leq   g  C \bold C^{2n+2} \frac{1}{\rho_n} 
   \rho_n^{  \frac{\mu}{2}} \leq g C   \frac{1}{\rho_n} 
   \rho_n^{  \frac{\mu}{4}}, 
\end{align} 
where we use \eqref{dorm2}. 
Notice that Eq.\ \eqref{dedo} implies that there is a constant $ C$ such that 
\begin{align}\label{dedo1}
\norm{\frac{1}{ H^{(n), \theta}-z}  \overline{P_{i }^{(n),\theta}} }  \leq C  \boldsymbol{C}^{n+1} \frac{1}{\rho_n}, 
\end{align} 
for every $  z     \in   B^{(1)}_i \setminus \mathcal C_m(\lambda_i^{(n)}  )    $.
Moreover,  \cite[Theorem 2.6]{bdh-res} implies that there is a constant $C$ such that
\begin{align}\label{resmain}
  \norm{\frac{1}{H^\theta - z}}  \leq C \boldsymbol{C}^{n+1} \frac{1}{ \rho_n^{1 + \mu/4}},
\end{align}
for every $z \in B^{(1)}_i  \setminus \mathcal{C}_m(\lambda_i^{(n)}- \rho_n^{1+ \mu/4} e^{-i\nu})$  and 
\begin{align}\label{resmainnn}
  \norm{\frac{1}{H^\theta - z}}  \leq C \boldsymbol{C}^{n+1} \frac{1}{{ \rm dist} \Big ( z, \mathcal{C}_m(\lambda_i  \Big ) },
\end{align}
for every $z \in B^{(1)}_i  \setminus \mathcal{C}_m(\lambda_i -  2\rho_n^{1+ \mu/4} e^{-i\nu})$.

\section{Proof of the main result}
\label{sec:proof-mainresult}

In the remainder of this work we provide the proofs of the   main result  Theorem~\ref{FKcor}. This section has three parts: In
Section~\ref{sec:proof-prelim}, we derive a preliminary formula for the
scattering matrix coefficients (see Theorem \ref{intker} below). This formula
together with several technical ingredients provided in 
Sections~\ref{sec:proof-techingredients} and \ref{kin} will pave the way for the proofs of
the main results given in Section~\ref{main}.

\subsection{Preliminary scattering formula}
\label{sec:proof-prelim}
In Theorem~\ref{intker} below we derive a preliminary formula for scattering
processes with one incoming and outgoing asymptotic photon. A  related formula
was already employed in \cite{spohnrad}. In order to derive it rigorously we
need several properties of the asymptotic creation and annihilation operators. The necessary properties are collected  in
Lemma~\ref{lemmadiff}. They have already been proven for a range of models in
several works \cite{fau1,fau2,fau3,fgs1,fgs2,rgk,rk,rk2,derezinski,bkz}. For
convenience of the reader we provide a self-contained proof of
Lemma~\ref{lemmadiff} in the Appendix~\ref{app:welldef}.

\begin{lemma}
\label{lemmadiff}
Let $\Psi\in \mathcal K \otimes D(H_f^{1/2})$ and $h,l\in\mathfrak h_0$.  The asymptotic
creation and annihilation operators $a_\pm^*, a_\pm$ defined in
Definition~\ref{defasymptop} have the following properties:
\begin{enumerate}[(i)]
    \item The limits $a^\#_\pm(h)\Psi=\lim_{t\to\pm\infty}a^\#_t(h)\Psi$ exist,   where $ a^\# $ stands for $a$ or $a^*$. 
    \item The next equalities holds true:
        \begin{align}
            \label{a_+}
            a_+(h)\Psi=a(h)\Psi - ig\int_0^\infty \mathrm{d}s\, 
         e^{isH}
        \langle h_s,f\rangle_2\,
            \sigma_1 e^{-isH}
            \Psi,
            \\
            \label{a_-}
            a_-(h)\Psi=a(h)\Psi + ig\int^0_{-\infty} \mathrm{d}s\, 
        e^{isH}
        \langle h_s,f\rangle_2\,
            \sigma_1
            e^{-isH} \Psi.
        \end{align} 
		We point out to the reader that   the integrals  above are convergent since it can be shown by    integration by parts   that there is constant $C$ such that $| \langle h_s,f\rangle_2 |\leq C/(1+s^2)$ for $s\in\R$ (see \eqref{eq:stat-phase} below).
    \item The following pull-through formula holds true:
        \begin{align}
        \label{inta}
             e^{-isH} a_- (h)^*\Psi= a_- (h_s)^* e^{-isH}\Psi.
        \end{align}
    \item The equality $a_\pm(h)\Psi_{\lambda_0}=0$ holds true, i.e.,
        $\Psi_{\lambda_0}\in\mathcal K^\pm$.
    \item The following commutation relation holds:
        $\langle a_\pm(h)^*\Psi_{\lambda_0}, a_\pm(l)^*\Psi_{\lambda_0}\rangle
        =\langle h, l\rangle_2 \| \Psi_{\lambda_0}    \|^2 $.
    \item There is a finite constant $C(h)>0$ such that for all 
        $t\in\R$
        \begin{align}
            \norm{a_t(h)^*(H_f+1)^{-\frac{1}{2}}} ,
            \norm{a_t(h)(H_f+1)^{-\frac{1}{2}}} \leq C(h) .
        \end{align}
\end{enumerate}
\end{lemma}
\begin{definition}
\label{def:fourier-distri}
Let $\mathit S(\R,\C)$ denote the Schwartz space of functions with rapid
  decay.  For all $u\in\mathit
S(\R,\C)$, we define the Fourier transform of a function and its inverse 
\begin{align}
\mathfrak{F}[u](x):= \int_\R \mathrm{d}s \,  u(s)e^{-isx},
\qquad
\mathfrak{F}^{-1}[u](x):= (2
\pi)^{-1}\int_\R \mathrm{d}s \,  u(s)e^{isx} .
\end{align}
Note the factor $(2\pi)^{-1}$ which is not  the normalization factor of     the standard definition of the
inverse Fourier transform. However, it is  convenient in our
 setting  (see e.g.\ \cite{reedsimon1}). 
\end{definition}
\begin{theorem}[Preliminary Scattering Formula]
   \label{intker}
   For $h,l\in \mathfrak{h}_0$, the two-body transition matrix coefficient
   $T(h,l)$ defined in \eqref{eq:Tmatrix} fulfills
\begin{align}
    T(h,l)=  \lim\limits_{t\to - \infty}\int \mathrm{d^3}k \mathrm{d^3}k' \,
\overline{h(k)} l(k') \delta(\omega(k)-\omega(k')) T_t(k,k')
\label{T}
\end{align}
for the integral kernel
\begin{align}
\label{intkernel}
T_t(k,k')=-2\pi i g f(k)  \norm{\Psi_{\lambda_0}}^{-2}{\langle \sigma_1 \Psi_{\lambda_0},
a_t(k')^* \Psi_{\lambda_0}\rangle}.
\end{align}
\end{theorem}
The integral in \eqref{T}
is to be understood as
\begin{align}
T(h,l)= -2\pi ig \norm{\Psi_{\lambda_0}}^{-2}\bigg\langle 
 \sigma_1  \Psi_{\lambda_0}, a_-(W)^* \Psi_{\lambda_0} \bigg\rangle 
 \label{eq:Tprecise1}
\end{align}
for $W\in\mathfrak{h}_0$ given by
\begin{align}
\label{def:W1st}
\R^3\ni k\mapsto W( k):=|k|^2 l(k) \int\mathrm{d}\Sigma \, \overline{h(|k|,\Sigma)}f(|k|,\Sigma) 
\end{align}
using spherical coordinates $k=(|k|,\Sigma)$ with $\Sigma$ being the solid
angle. 
\begin{proof}
Let $h,l\in \mathfrak{h}_0$. Thanks to Lemma~\ref{lemmadiff} (i) and the fact
that the ground state $\Psi_{\lambda_0}$ lies in $\mathcal D(H)= \mathcal
K\otimes \mathcal D(H_f)$, c.f.\ \cite[Theorem 1]{spohnspin} and
Proposition~\ref{thm:Hsa},
the transmission matrix coefficient given in \eqref{eq:Tmatrix}, i.e.,
\begin{align}
    T(h,l)= S(h,l)- \left\langle h,l \right\rangle_2
    = \norm{\Psi_{\lambda_0}}^{-2}\langle
    a_+(h)^*\Psi_{\lambda_0},a_-(l)^*\Psi_{\lambda_0} \rangle
    - \left\langle h,l \right\rangle_2
    \label{eq:T1}
\end{align}
is well-defined. 
Lemma~\ref{lemmadiff} (iv) and (v) implies that
\begin{align}
    (\ref{eq:T1})= \norm{\Psi_{\lambda_0}}^{-2}\langle
    [a_+(h)^*-a_-(h)^*]\Psi_{\lambda_0},a_-(l)^*\Psi_{\lambda_0} \rangle.
    \label{eq:T2}
\end{align}
Using Lemma~\ref{lemmadiff} (ii), we obtain
\begin{align}
    (\ref{eq:T1})
    =-i g\norm{\Psi_{\lambda_0}}^{-2} \int_{-\infty}^\infty \mathrm{d}s 
    \langle
    \Psi_{\lambda_0},e^{isH}\sigma_1 e^{-isH} a_-(l)^*\Psi_{\lambda_0} \rangle
    \langle h_s,f\rangle_2
    .
    \label{eq:T3}
\end{align}
Finally, we use Lemma~\ref{lemmadiff} (iii) to get
\begin{align}
    (\ref{eq:T1}) \notag 
    &=-i g\norm{\Psi_{\lambda_0}}^{-2}\int_{-\infty}^\infty \mathrm{d}s
    \left\langle
    e^{-isH} \Psi_{\lambda_0},\sigma_1 a_-(l_s)^*e^{-isH} \Psi_{\lambda_0} \right\rangle 
    \langle h_s,f \rangle_2
    \\
    &=-i g\norm{\Psi_{\lambda_0}}^{-2}\int_{-\infty}^\infty \mathrm{d}s 
    \left\langle
     \sigma_1  \Psi_{\lambda_0}, a_-(l_s)^*\Psi_{\lambda_0} \right\rangle 
    \langle h_s,f\rangle_2
.
\end{align}
We insert the definition of the asymptotic creation operator
in \eqref{asymptop} to find
\begin{align}
(\ref{eq:T1})
=-i g \norm{\Psi_{\lambda_0}}^{-2} \int_{-\infty}^\infty \mathrm{d}s
\lim_{t\to-\infty}\left\langle
    \sigma_1 { \Psi_{\lambda_0}}, a_t(l_s)^* { \Psi_{\lambda_0}} \right\rangle
\langle h_s,f\rangle_2.
\label{eq:ds-integrand}
\end{align}
Next, it is possible to interchange the $\mathrm{d}s$ integral and the limit
$t\to-\infty$. This can be seen as follows.  A two-fold partial 
integration  implies that there is a constant $C$ such that,   for all $s\in\R$, we get
\begin{align}
  |\langle h_s,f\rangle_2|  \leq C  \frac{1}{1 + |s|^2} 
   .
    \label{eq:pint}
\end{align}  
By applying
Lemma~\ref{lemmadiff} (vi), we infer that there is a
finite constant $C_{\eqref{eq:const-a}}(l)>0$ such that for all 
$s\in\R$
\begin{align}
    |\langle \sigma_1 { \Psi_{\lambda_0}}, a_t(l_s)^* { \Psi_{\lambda_0}} \rangle |
    &\leq
    \|\sigma_1\Psi_{\lambda_0}\|\, \|a_t(l_s)^* (H_f+1)^{-\frac
    12}\|\,\|(H_f+1)^{\frac12}{ \Psi_{\lambda_0}}\|
    \notag
    \\
    &\leq 
    C_{\eqref{eq:const-a}}(l)\| { \Psi_{\lambda_0}}  \|\| { \Psi_{\lambda_0}}  \|_{H_f}
    \label{eq:const-a}
\end{align}
holds true. Both estimates, \eqref{eq:pint} and \eqref{eq:const-a}, give an
integrable bound of the ds-integrand in 
\eqref{eq:ds-integrand} that is uniform in $t$.
Hence, by dominated convergence, we have the equality
\begin{align}
  (\ref{eq:T1})   &
=-i g \norm{\Psi_{\lambda_0}}^{-2} \lim_{t\to-\infty}\int_{-\infty}^\infty \mathrm{d}s
\left\langle
    \sigma_1    { \Psi_{\lambda_0}   }  , a_t(l_s)^* 
     { \Psi_{\lambda_0}}   \right\rangle
\langle h_s,f\rangle_2
\notag\\
 & = 
-i g \norm{\Psi_{\lambda_0}}^{-2} \lim_{t\to-\infty}  e^{-i t \lambda_0 } \int_{-\infty}^\infty \mathrm{d}s   
      \left\langle
    e^{- i t H }\sigma_1    { \Psi_{\lambda_0}}  , a( l_{s+ t} )^* 
    { \Psi_{\lambda_0}} \right\rangle
\langle h_s,f\rangle_2,
    \label{ddd0}
\end{align}
where in the last step we have inserted definition \eqref{asymptop} and
exploited the ground state property 
  \eqref{gsprop}.

In order to rewrite this integral in form of \eqref{T}-\eqref{intkernel}, or
more precisely, \eqref{eq:Tprecise1}-\eqref{def:W1st}, we
shall use the following approximation argument.  Let 
\begin{align}
    \label{eq:H0}
   \mathcal H_0:=\mathcal K\otimes
\mathcal F_{\text{fin}}[\mathfrak h_0] 
\end{align}
be the set of states with only finitely
many bosons, i.e., 
\begin{align}
\label{def:denseF}
    \mathcal F_{\text{fin}}[\mathfrak h_0]:= \Big \{ \Psi=(\psi^{ (n)})_{n\in\N_0}\in
        \mathcal F[\mathfrak h] \,\big|\, 
        \exists  N  \in \N_0   : &  
\psi^{ (n)}=  0 \, \forall n\geq N, \\ \notag &  \forall n \in \mathbb{N} :
   \psi^{ (n)} \in C_c^{\infty}(\mathbb{R}^{3n } \setminus \{ 0 \}, \mathbb{C})    \Big \}. 
\end{align}
Note that $\mathcal H_0$ is a dense subset of $\mathcal H$ with respect to the
norm in $ \mathcal{H} $ and it is dense in the domain of $H_f$ with respect to
the graph norm of the operator $H_f$ defined by
$\norm{\cdot}_{H_f}:=\norm{H_f\cdot}+\norm{\cdot}$.  Hence,  for $t\in\R$,
 there are  sequences $(\Psi_m)_{m\in\N}$,   $(
\Phi_m^t)_{m\in\N}$ in $\mathcal H_0$ with
$\norm{\Psi_m-\Psi_{\lambda_0}}_{H_f}\to 0$, as $m \to\infty$, and
$\norm{ \Phi_m^t-    e^{- i t H }\sigma_1    \Psi_{\lambda_0}
}  \to 0$, as $m \to\infty$. Then, Lemma \ref{lemma:standardest1}, applied in
the same fashion as 
in \eqref{eq:const-a},  implies
that
\begin{align}
\lim_{m \to \infty}   \left\langle
     \Phi_m^t  , a( l_{s+ t} )^* 
     \Psi_{m} \right\rangle =  \left\langle
    e^{- i t H }\sigma_1    { \Psi_{\lambda_0}}  , a( l_{s+ t} )^* 
    { \Psi_{\lambda_0}} \right\rangle,
\end{align}
uniformly  in   $s$. 
 Thanks to the bound
\eqref{eq:pint}, we may apply 
dominated convergence theorem to conclude that 
\begin{align} 
\lim_{m \to \infty }  \int_{-\infty}^\infty \mathrm{d}s &   
      \left\langle
     \Phi_m^t   , a( l_{s+ t} )^* 
     \Psi_{m} \right\rangle 
    \langle h_s,f\rangle_2
  =   \int_{-\infty}^\infty \mathrm{d}s   
      \left\langle
    e^{- i t H }\sigma_1    { \Psi_{\lambda_0}}  , a( l_{s+ t} )^* 
    { \Psi_{\lambda_0}} \right\rangle
    \langle h_s,f\rangle_2
    .
\label{ddd}
\end{align}
Now, we study the integrals in the left hand side of Eq.\ \eqref{ddd}. The
advantage of the sequences  $(\Psi_m)_{m\in\N}$,   $(  \Phi_m^t)_{m\in\N}$  is
that they allow to use point-wise annihilation operators in the following
manner: \begin{align}\label{ddd1}
\int_{-\infty}^\infty \mathrm{d}s &   
      \left\langle
     \Phi_m^t  , a( l_{s+ t} )^* 
     \Psi_{m} \right\rangle    
\langle h_s, f\rangle_2  \\ \notag & = 
    \int_{-\infty}^\infty \mathrm{d}s \int d^3 k' e^{-is \omega(k')}   e^{-it \omega(k')}   l(k')     
      \left\langle
    a(k')   \Phi_m^t   ,  
     \Psi_{m} \right\rangle 
\int
    \mathrm{d^3}k \,\overline{h(k)} f(k)    e^{is\omega(k)} 
    \\ \notag &  = \int_{- \infty}^{\infty} ds \Big [ \Big ( \int_{- \infty}^{\infty }    \mathrm{d}r  \, e^{isr} \Theta(r) u(r) \Big ) \Big ( \int_{- \infty}^{\infty}  \mathrm{d}r'\, 
    e^{-isr'}  \Theta(r')  v_m^t (r') \Big ) \Big ] ,  
\end{align}   
where $\Theta$ is the Heaviside function and we use {spherical coordinates and} the abbreviations
\begin{align}\label{ddd2}
u(r):=r^2\int \mathrm{d}\Sigma \, \overline{h(r,\Sigma)} f(r,\Sigma)
\quad \text{and} \quad 
 v_m^t (r'):= e^{-it r' }  r'^2\int \mathrm{d}\Sigma' \, l(r',\Sigma') \left\langle
a(r',\Sigma')  \Phi^t_m , \Psi_m\right\rangle.
\notag
\end{align}  
By  definition,  $v^t_m $  and $u$
belong to $C_c^{\infty}( \mathbb{R} \setminus \{ 0 \} ) $ so that
the integrals with respect to  $r$ and $r'$ above can be regarded as Fourier
transform,  introduced in Definition \ref{def:fourier-distri},  i.e.,
\begin{align} 
{\eqref{ddd1}}
= \int_{-\infty}^\infty \mathrm{d}s \,
\overline{\mathfrak F \left[ \overline{\Theta  u}   \right]} (s)
{\mathfrak F \left[\Theta  { v_m^t} \right] (s)}    
\end{align}
holds true.
 Plancherel's identity yields for all $t\in\R$
\begin{align} 
\label{ddd3} 
     (\ref{ddd1})  & = 2\pi \int_{-\infty}^{\infty}   dr' \,
   {   \Theta  u } \Theta  v_m^t (r') 
     \notag \\
     &  =
    2\pi  \int_{0}^\infty  dr' \, {r'}^2  \int \mathrm{d}\Sigma
    \,\overline{ h(r',\Sigma) } f(r',\Sigma) e^{-it
    r' }  {r'}^2 \int \mathrm{d}\Sigma' \, l( r',\Sigma') \left\langle a(r',\Sigma') 
        \Phi_m^t,
        \Psi_m\right\rangle \notag \\ & =  2\pi  \left\langle a(  W_t )
         \Phi_m^t,
    \Psi_m\right\rangle = 2\pi\left\langle  \Phi_m^t, a(  W_t )^*  \Psi_m\right\rangle
\end{align}
where we have used the definition of $W$ in \eqref{def:W1st} and the definition
\eqref{asymptop}, in particular, the notation $W_t(k)=W(k)e^{-it\omega(k)}$.
Using Lemma \ref{lemma:standardest1}, applied in
the same fashion as in
 \eqref{eq:const-a}, allows to carry   out the limit $m\to\infty$ which
results in
\begin{align}
    \eqref{ddd}
    =
    \lim_{m\to\infty}\eqref{ddd3} = 
    2\pi\left\langle e^{ - itH}\sigma_1\Psi_{\lambda_0} 
    ,
     \ a(  W_t )^*
    \Psi_{\lambda_0}\right\rangle.
\end{align}
This together with \eqref{ddd0}  and Lemma~\ref{lemmadiff}     guarantees
\begin{align}
    \eqref{eq:T1} &= 
-i g \norm{\Psi_{\lambda_0}}^{-2} \lim_{t\to-\infty}  e^{-i t \lambda_0 }
 2\pi\left\langle
    e^{-itH}\sigma_1\Psi_{\lambda_0}, a(  W_t )^*
    \Psi_{\lambda_0}\right\rangle
    \\
    &=
-2\pi i g \norm{\Psi_{\lambda_0}}^{-2} \lim_{t\to-\infty}  
 \left\langle
    \sigma_1\Psi_{\lambda_0}, a_t(  W )^*
    \Psi_{\lambda_0}\right\rangle
    \\
    &=
-2\pi i g \norm{\Psi_{\lambda_0}}^{-2} 
 \left\langle
    \sigma_1\Psi_{\lambda_0}, a_-(  W )^*
    \Psi_{\lambda_0}\right\rangle ,
\end{align}
which concludes the proof.
\end{proof}

\subsection{Technical ingredients}
\label{sec:proof-techingredients}
Here, we derive some technical results which will be applied in the proof of
the main results in Section \ref{main}.  The statements in this
section will mostly be formulated without motivation, however, their importance
will become clear later in the proofs of the main results. 

\subsubsection{General results}
\begin{lemma}
\label{lemma:orth}
For $n\in\N$ and  $\theta \in\mathcal S$, we have 
\begin{align}
P^{(n),\theta}_{0}\sigma_1P^{(n),\theta}_{0}=0 .
\end{align}
\end{lemma}
The statement has already been proven in \cite[Lemma 2.1]{bbkm}.

 Next, we prove a representation formula of the evolution operator similar
to the Laplace transform representation (see, e.g., \cite{bach}).
\begin{lemma}
\label{laplace}
For $\epsilon>0$ and sufficiently large $R >0$, we consider the concatenated contour
$\Gamma(\epsilon,R):=\Gamma_{-}(\epsilon,R)\cup
\Gamma_{c}(\epsilon)\cup \Gamma_{d}(R)$ (see Figure \ref{fig:CurveGamma}),
where
\begin{align}
\Gamma_{-}(\epsilon,R)&:=[-R, \lambda_0 -\epsilon]\cup [\lambda_0 +\epsilon,R
], 
\notag \\
\Gamma_{d}(R)&:=\left\lbrace -R -ue^{i\frac{\nu}{4}} :u\geq 0 \right\rbrace
\cup \left\lbrace R +ue^{-i\frac{\nu}{4}} :u\geq 0 \right\rbrace ,
\notag \\
\Gamma_{c}(\epsilon)&:=\left\{ \lambda_0 -\epsilon e^{-it}: t\in [0,\pi]
\right\}.
\label{Gamma-parts}
\end{align}  
The  orientations of the contours in \eqref{Gamma-parts} are given by the arrows depicted in Figure
\ref{fig:CurveGamma}.   
Then, for all analytic vectors $\phi,\psi\in\mathcal H$
(analytic in a --  connected --  domain containing   $0$)
    and $t>0$ the following identity holds true:
\begin{align}
\left\langle \phi , e^{-itH} \psi  \right\rangle =
\frac{1}{2\pi i}\int_{\Gamma(\epsilon,R)}\mathrm{d}z\, e^{-itz} \left\langle \psi^{\overline{\theta}} , \left( H^\theta-z  \right)^{-1} \phi^\theta \right\rangle  .
\label{eq:laplace}
\end{align}
\end{lemma}
\begin{figure}[h] 
\centering
\includegraphics[width=\textwidth]{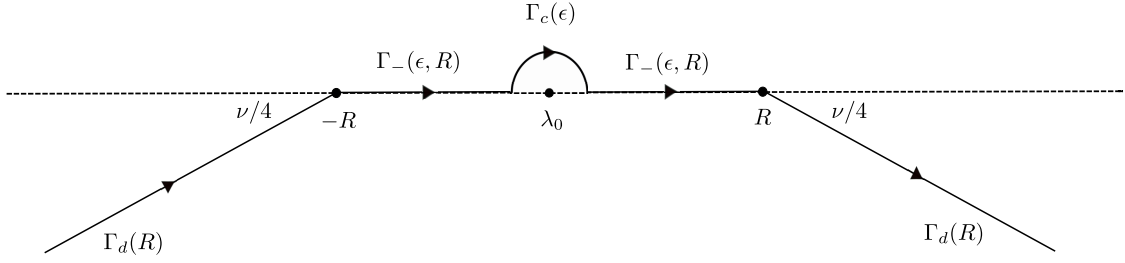}
  \caption{An illustration of the contour
  $\Gamma(\epsilon,R):=\Gamma_{-}(\epsilon,R)\cup \Gamma_{c}(\epsilon)\cup
  \Gamma_{d}(R)$. }
    \label{fig:CurveGamma}       
\end{figure}
\begin{proof}
    Let $t>0$ and $\epsilon>0$. 
    We define a contour $\hat \Gamma(\epsilon):= \R+i\epsilon$ with a
    mathematical negative orientation if the contour were closed in the lower
    complex plane.
    As an application of the residue theorem closing the contour in the
    lower complex plane, we observe for all $E\in\R$
    \begin{align}
        \frac{1}{2\pi i} \int_{\hat \Gamma(\epsilon)}\mathrm{d}z\, \frac{
        e^{-itz} }{it(E-z)^2}= e^{-itE}
    \end{align}
    holds true.
    Thanks to the spectral theorem we may write for all $\psi\in\mathcal H$
    \begin{align}
        \left\langle \psi , e^{-itH} \psi \right\rangle  
        =\int_{\sigma(H)}\left\langle \psi , \mathrm{d}P_E \psi \right\rangle e^{-itE} 
        = 
        \frac{1}{2\pi i}
        \int_{\sigma(H)}
        \int_{\hat \Gamma(\epsilon)}\mathrm{d}z\,
        \left\langle \psi ,
            dP_E \psi \right\rangle 
        \frac{
            e^{-itz} }{it (
        E-z  )^{2}}.
        \label{eq:interchange}
    \end{align}
    Next, we may interchange the order of the integrals by the Fubini-Tonelli
    Theorem since the following integral 
    is finite: 
    \begin{align}
        \int_{\sigma(H)}
        \left\langle \psi ,
        dP_E \psi \right\rangle
        \int_{\hat \Gamma(\epsilon)}\mathrm{d}z\,
        \left|
        \frac{
            e^{-itz} }{it (
        E-z  )^{2}}
        \right|
     \leq
        \frac{e^{t\epsilon}}{t}
        \int_{\sigma(H)}
        \left\langle \psi ,
        dP_E \psi \right\rangle
        \int_{-\infty}^\infty dx\,
        |
        x-i\epsilon |^{-2}
        <\infty.
    \end{align}
Hence, after the interchange we may apply the spectral theorem
    again to find
    \begin{align}
        \eqref{eq:interchange}
        &=
        \frac{1}{2\pi i}
        \int_{\hat \Gamma(\epsilon)}\mathrm{d}z
        \int_{\sigma(H)}
        \,
        \left\langle \psi ,
            dP_E \psi \right\rangle 
        \frac{
            e^{-itz} }{it (
        E-z  )^{2}}
        =
        \frac{1}{2\pi i}
        \int_{\hat \Gamma(\epsilon)}\mathrm{d}z\,
        \frac{e^{-itz}}{it}
        \left\langle \psi ,
        \frac{1}{(H-z)^2} \psi \right\rangle .
    \end{align}
    Exploiting the polarization identities we recover for all
    $\psi,\phi\in\mathcal H$ the identity
    \begin{align}
        \left\langle \psi , e^{-itH} \phi \right\rangle  
        =
        \frac{1}{2\pi i}
        \int_{\hat \Gamma(\epsilon)}\mathrm{d}z\,
        \frac{e^{-itz}}{it}
        \left\langle \psi ,
        \frac{1}{(H-z)^2} \phi \right\rangle.
        \label{eq:sandwich}
    \end{align}
   
The fact that the family $ H^{\theta} $ is an analytic family of type $A$
implies that the operator valued function 
\begin{align}\label{corr1}
\theta \mapsto \frac{ 1}{H^{\theta} - z }
\end{align}
is analytic for all $z$ in the resolvent set of
$H^{\theta}$. A detailed and self-contained   exposition of this topic is
presented in  \cite[Section 7]{bdh-res}.  It is straight forward to prove that
for real $ \theta $  
\begin{align}\label{corr2prima}
\frac{ 1}{H^{\theta} - z } = U^\theta \frac{ 1}{   H  - z } ( U^\theta
)^{-1}.
\end{align}
 For complex $ \theta$, however,  this expression is not necessarily correct (due to a problem
of domains of unbounded operators).  Nevertheless,  Eqs.\ \eqref{corr1} and \eqref{corr2prima}
imply that the function 
\begin{align}\label{corr2}
\theta \mapsto  \left\langle \psi^{\overline\theta} ,
        \frac{1}{(H^\theta-z)^2} \phi^\theta \right\rangle,
\end{align}
  where $   \phi^\theta   = U^{\theta}  \phi,  \:   \psi^{\overline{\theta}}   = U^{\overline{\theta}}  \psi   , $  
is analytic and it coincides with  $ \left\langle \psi ,
        \frac{1}{(H-z)^2} \phi \right\rangle $ for real $\theta$, because in
        this case $U^\theta $ is unitary. Hence, we conclude that 
\begin{align} \label{corr3}
\left\langle \psi^{\overline\theta} ,
        \frac{1}{(H^\theta-z)^2} \phi^\theta \right\rangle = 
        \left\langle \psi ,
        \frac{1}{(H-z)^2} \phi \right\rangle
\end{align}
  for every $\theta$ in a connected (open) domain containing $0$ such that \eqref{corr2} is analytic in this domain. We obtain: 
    \begin{align}
        \eqref{eq:sandwich}=
        \frac{1}{2\pi i}
        \int_{\hat \Gamma(\epsilon)}\mathrm{d}z\,
        \frac{e^{-itz}}{it}
        \left\langle \psi^{\overline\theta} ,
        \frac{1}{(H^\theta-z)^2} \phi^\theta \right\rangle
        \label{eq:dil-h-sandwich}
    \end{align}
Eqs.\ \eqref{spectrum}  and \eqref{eq:impartres}   imply that the only spectral point of $ H^{\theta} $ on the real line is $\lambda_0^{\theta}$ and all other spectral points have strictly negative imaginary part. Therefore, the operator valued function  
    \begin{align}
        \label{eq:laplace-ana}
        A \cup  { \mathbb{C}^{+} }   \ni z \mapsto
        \frac{1}{H^\theta-z},
    \end{align}
  where $  \mathbb{C}^+ =  \{  x+ i y |  x \in \mathbb{R}, y >  0 \} $,   is  analytic.
  Moreover,  for $ R \geq e_1 + \delta = 2 e_1 $, $\Gamma_d (R)$ is contained in the  region
  $A $, and hence, it follows from \eqref{resA}  that there is a constant  $C $ such that
    \begin{align}
        \label{eq:laplace-ana1} 
        \norm{\frac{1}{H^\theta-z}}\leq
        \frac{C}{| z - e_1 |}
           \qquad \forall
        z\in \Gamma_d.
    \end{align} 
    Due to the analyticity, we may deform the integration contour from $\hat
    \Gamma(\epsilon)$ to $\Gamma(\epsilon,R)$ which gives:
    \begin{align}
        \eqref{eq:dil-h-sandwich}
        =
        \frac{1}{2\pi i}
        \int_{\Gamma(\epsilon, R)}\mathrm{d}z\,
        \frac{e^{-itz}}{it}
        \left\langle \psi^{\overline\theta} ,
        \frac{1}{(H^\theta-z)^2} \phi^\theta \right\rangle .
    \end{align}
    Now we observe that the integrand on the right-hand side features an exponential decay for large $|\Re
    z|$ thanks to the factor $e^{-itz}$ in the integrand and the definition of
    $\Gamma_d(\epsilon, R)$. In particular, the decay in $|z|$ provided by the
    resolvent, i.e., bound \eqref{eq:laplace-ana1}, is not necessary anymore to make the integral converge. We may
    therefore perform an integration by parts.
    Note that, for $z$ in  $A \cup \mathbb{C}^{+}$,  we have
    \begin{align}
        \frac{d}{d z}
        \left\langle \psi^{\overline\theta} ,
        \frac{1}{(H^\theta-z)} \phi^\theta \right\rangle
        =
        \left\langle \psi^{\overline\theta} ,
        \frac{1}{(H^\theta-z)^2} \phi^\theta \right\rangle
    \end{align}
    which is implied by the resolvent identity 
    \begin{align}
       \left\langle \psi^{\overline\theta} ,
        \frac{1}{(H^\theta-(z+  u))} \phi^\theta \right\rangle 
        -
       \left\langle \psi^{\overline\theta} ,
        \frac{1}{(H^\theta-z)} \phi^\theta \right\rangle 
        =
       \left\langle \psi^{\overline\theta} ,
       \frac{1}{(H^\theta-(z+  u))}  u\frac{1}{(H^\theta-z)} \phi^\theta \right\rangle .  
    \end{align} 
    Moreover, the boundary terms of the
    partial integration resulting from the piece-wise concatenation of contours, i.e., 
   $\Gamma(\epsilon,R)=\Gamma_{-}(\epsilon,R)\cup
\Gamma_{c}(\epsilon)\cup \Gamma_{d}(R)$,
    cancel and the ones
    at $|\Re z|\to\infty$ vanish because of the exponential decay. In conclusion,
    the identity
    \begin{align}
        \eqref{eq:dil-h-sandwich}
        =
        \frac{1}{2\pi i}
        \int_{\Gamma(\epsilon, R)}\mathrm{d}z\,
        e^{-itz}
        \left\langle \psi^{\overline\theta} ,
        \frac{1}{H^\theta-z} \phi^\theta \right\rangle 
    \end{align}
    holds true which proves the claim.
\end{proof}

\subsection{Key ingredients}\label{kin}

The next definition is motivated by a simple
geometric argument  which we give
in the following for the convenience of the reader: take a
cone of the form $\mathcal C_m ( \lambda_0^{(n)} - x  e^{-i \nu }) $, $x > 0$,
where $m$ is a fixed (arbitrary) number  greater or equal than $4$. Although $m$ is arbitrary, our estimates and constants depend on it.
The distance between the vertex of the cone and the intersection of the line  $
\lambda_0^{(n)} - i x \sin(\nu) + \mathbb{R} $ with the cone is $$  \sqrt{ \Big
( \frac{2 x \sin(\nu) }{ \tan\big (  (1 - 1/m)\nu \big )}\Big ) ^2 +  (2 x
\sin(\nu))^2   } \leq 4 x  \frac{\sin(\nu)}{  \sin \big (  (1 - 1/m)\nu \big )}
\leq  8 x. $$ To obtain the last inequality we use the sum of angles formula for
$ \sin(\nu) $, writing $   \nu = (\nu - \nu/m) + \nu/m  $. Then, we have that
the   distance between $ \lambda_0^{(n)}  $  and the  line segment described
above is smaller than $  8 x $.  
\begin{definition}\label{epsilonn} 
For every $n \in \mathbb{N}$,
we define 
\begin{align}
\label{eq:epsn}
\epsilon_n  : =  20 \rho_n^{1 + \mu/ 4 }  .
\end{align}
It follows from \eqref{rhoneuno} and \eqref{ground1} that for every 
 $ n \in \mathbb{N} $ 
\begin{align}
D(\lambda_0, 2  \epsilon_n ) \subset B_0^{(1)} . 
\end{align}
The geometric argument given above together with
$ |\lambda_0^{(n)} - \lambda_0 | \leq  10^{-2}   \rho_n^{1+ \mu/2} $ (see Definition \ref{gggg} and \eqref{ground}) yields that, for all $n\in\N$ 
and a fixed (arbitrary)  $m\geq 4$, 
\begin{align}\label{carita}
\mathcal C_m ( \lambda_0^{(n)} - 2 \rho_{n}^{1+ \mu/4}  e^{-i \nu } ) \cap
 \Big ( \overline{\mathbb{C}^+} + \lambda_0^{(n)} - i 2 \sin(\nu) \rho_n^{1 + \mu /4 } \Big ) \subset    D(\lambda_0,   \epsilon_n )  \subset 
 D(\lambda_0,  2 \epsilon_n ) \subset B_0^{(1)} 
\end{align}
 and
\begin{align}\label{caritaputa}
\mathcal C_m ( \lambda_0 - 2 \rho_{n}^{1+ \mu/4}  e^{-i \nu } ) \cap
 \Big ( \overline{\mathbb{C}^+} + \lambda_0 - i 2 \sin(\nu) \rho_n^{1 + \mu /4 } \Big ) \subset    D(\lambda_0,   \epsilon_n )  \subset 
 D(\lambda_0,  2 \epsilon_n ) \subset B_0^{(1)} .
\end{align}
Note that \eqref{ground} 
and the fact that  $\lambda_0 \in\R $ imply that 
\begin{align}\label{Upsilon}
\Im \lambda_0^{(n)} - 2 \sin(\nu) \rho_n^{1 + \mu /4 } \leq 2 g \rho_n^{1+ \mu/2} 
-   2 \sin(\nu) \rho_n^{1 + \mu /4 }  < 0, \hspace{1cm} \forall n \in \mathbb{N}, 
\end{align} 
for small enough $g$ (see Definition 4.3  in \cite{bdh-res}). 
Eq.\ \eqref{carita}  implies that for every  
   $n \in \N$  
\begin{align}
 \Gamma_c(\epsilon_n) \subset   B_0^{(1)} \setminus   \mathcal C_m ( \lambda_0^{(n)} - 2 \rho_{n}^{1+ \mu/4}  e^{-i \nu } ).    
\end{align}
\end{definition}
\begin{lemma}\label{kkk1}
 For all $n\in\N$,  
a fixed (arbitrary)  $m\geq 4$ 
and $\theta\in\mathcal S$, there is a constant   $ C $ 
(that depends on $m$) 
such that
\begin{align}\label{nn0}
 \Big \| \frac{1}{H^{\theta}-z} \sigma_1 
 \Psi_{\lambda_0}^{\theta}    \Big \| \leq  C \boldsymbol{C}^{n+1} \frac{1}{\rho_n},
\end{align}
for every $ z \in  B_0^{(1)} \setminus  \mathcal{C}_{m}( \lambda_0^{(n)} -  \rho_n^{1+ \mu/ 4} e^{-i \nu}) $,  and hence, for every $ z \in  B_0^{(1)} \setminus  \mathcal{C}_{m}( \lambda_0 - 2 \rho_n^{1+ \mu/ 4} e^{-i \nu}) $, see \cite[Theorem 5.10]{bdh-res}. 
\end{lemma}
\begin{proof} We take   $ z \in  B_0^{(1)} \setminus  \mathcal{C}_{m}(  \lambda_0^{(n)}  - \rho_n^{1+ \mu/ 4} e^{-i \nu}) $ 
and recall  the definition 
$  \Psi^\theta_{\lambda_0} = P_0^{\theta}   \varphi_0 \otimes \Omega$. Then, Eq.\ \eqref{projj}  yields
\begin{align}\label{nn1}
\|  \Psi^\theta_{\lambda_0} -  P_0^{(n ),\theta} \otimes P_{\Omega^{(n , \infty)}}   \varphi_0 \otimes \Omega \| \leq  \rho_n^{   \mu / 2}.
\end{align} 
This  together with Eqs.\  \eqref{diffet}, 
\eqref{resmain},  \eqref{dorm2} and \eqref{projj2} implies that there is a constant $C$ such that 
(we use a
    telescopic sum argument) 
\begin{align}
 & \Big \|    \frac{1}{H^{\theta}-z} \sigma_1  
 \Psi_{\lambda_0}^{\theta}     -  
  \frac{1}{\tilde H^{(n), \theta}-z} \sigma_1 
  P_0^{(n), \theta} \otimes P_{\Omega^{(n , \infty)}}   \varphi_0 \otimes \Omega  \Big \|  
  \notag \\
  & \leq  \norm{\frac{1}{H^{\theta}-z} -\frac{1}{\tilde H^{(n), \theta}-z}}\norm{P_0^{(n ),\theta}} + \norm{\frac{1}{H^{\theta}-z}} \norm{\Psi^\theta_{\lambda_0} -  P_0^{(n ),\theta} \otimes P_{\Omega^{(n , \infty)}}   \varphi_0 \otimes \Omega}
  \notag \\
  &\leq
  C \boldsymbol{C}^{n+1} \frac{1}{\rho_n}  . \label{nn2t} 
\end{align}
 The fact (see Remark \eqref{R}) that 
\begin{align}\label{sss1}
 \Big ( \frac{1}{\tilde H^{(n), \theta}-z} \sigma_1 \Big ) 
  \Big ( P_0^{(n ),\theta} \otimes P_{\Omega^{(n , \infty)}} \Big)   \varphi_0 \otimes \Omega
  = \Big ( \big ( \frac{1}{ H^{(n), \theta}-z} \sigma_1 \big ) \otimes  P_{\Omega^{(n , \infty)}}  \Big )
   P_0^{(n ),\theta}      \varphi_0 \otimes \Omega  
\end{align}
guarantees that there is a constant $C$ such that
\begin{align}\label{nn3}
\Big \|   \frac{1}{\tilde H^{(n), \theta}-z} \sigma_1 
  P_0^{(n), \theta} \otimes P_{\Omega^{(n , \infty)}}   \varphi_0 \otimes \Omega  \Big \|   
\leq  \Big \|  \Big (  \overline{ P_0^{(n ), \theta} }  \frac{1}{ H^{(n), \theta}-z}  \Big ) \otimes   P_{\Omega^{(n , \infty)}}  \Big \|   \leq C \frac{\boldsymbol{C}^{n+1}}{ \rho_n }.
\end{align}
 Here, we use  Eqs.\ \eqref{dedo1}, \eqref{projj2} and Lemma \ref{lemma:orth}. 
\end{proof}
\begin{remark}\label{RRRR}  
Set $c \equiv c_g := {\rm Im} \lambda_1$. Notice that  there is a strictly positive  constant $\boldsymbol{c} $ (independent of $g$) with $c_g  \leq - g^2 \boldsymbol{c}   $, for small enough $g$ (see \eqref{eq:impartres}). Then, for  all real numbers $b> a$    and every $x  \in \{ 0, 1 \} $,  
\begin{align}
\label{rodos1}
&\int_{a}^{b} \mathrm{d}r\, \frac{1}{| r -  \lambda_1|^{1+ x}}   = \int_{a}^{b} \mathrm{d}r\, \frac{1}{g^{2(1+x)}} \frac{1}{ (  (r- {\rm Re} \lambda_1)/g^2)^2 +  (c/g^2)^2 )^{(1+ x)/2}   } \\ \notag &    \leq   \frac{1}{g^{2x}} \int_{(a - {\rm Re} \lambda_1  )/g^2}^{(b  - {\rm Re} \lambda_1  )/g^2} \mathrm{d}y\, \frac{1}{ ( y^2 +  \boldsymbol{c}^2  )^{(1+x)/2}}  \notag \\   &  \leq 2 \frac{1}{g^{2x}} \int_{0}^1 \mathrm{d}\tau\,  \frac{1}{ \boldsymbol{c}^{1+x}}      +2  \frac{1}{g^{2x}} \int_{1}^{  1 +   |   b  - {\rm Re} \lambda_1    |/g^2  +  | a  - {\rm Re} \lambda_1   |/g^2  }  \mathrm{d} \tau\,\frac{1}{\tau^{1+x}}  \notag
\\ &  \notag
\leq 
C  \begin{cases}    \frac{1}{ g^{2}  } , &  \text{if $x = 1$}, \\
|\log(g) |, &  \text{if $x = 0$},  \end{cases}
\end{align}
 where $C$ does not depend on $g$  (for small enough $g$). 
\end{remark}
\begin{lemma}\label{malhe}
Set $\mathcal{L}: = B^{(1)}_1  \cap \mathbb{R}$.  Then, for $g>0$ sufficiently small and $\theta \in \mathcal S$,  
\begin{align}\label{nooo}
\int_{\mathcal{L}}  \mathrm{d}r \,  \norm{ \frac{1}{ H^{\theta} - r   }  \sigma_1 \Psi^\theta_{\lambda_0}   } \leq C |\log g| , 
\end{align}
where $C$ is a constant that does not depend on $g$. 
\end{lemma}
\begin{proof}
Let $ \boldsymbol{c} $ be the constant introduced in \eqref{Im}, see   Remark \ref{RRRR}. The vertex of the cone $ \mathcal{C}_m( \lambda_1 - 2 \rho_{n} ^{1+ \mu/4}  e^{-i\nu})  $ belongs to the lower (open) half space of the complex plane if 
  $$ - g^2 \boldsymbol{c} +    2 \rho_{n} ^{1+ \mu/4} \sin(\nu) < 0  . $$ 
This is fulfilled if  
$$
n  >  \log \Big ( \frac{  g^2  \boldsymbol{c} }{   2 \sin(\nu) \rho_0^{1+ \mu/4}  }  \Big ) \frac{1}{ (1 + \mu/4) \log(\rho)  }.   
 $$  
 We fix $n_0 > 0$ to be the smallest integer number satisfying this  inequality. Then   
\begin{align}\label{noes}
n_0  \leq   \log \Big ( \frac{  g^2  \boldsymbol{c} }{   2 \sin(\nu) \rho_0^{1+ \mu/4}  }  \Big ) \frac{1}{ (1 + \mu/4) \log(\rho)  } + 1.    
 \end{align}  
For such $n_0$, the cone $ \mathcal{C}_m( \lambda_1 - 2 \rho_{n_0} ^{1+ \mu/4}  e^{-i\nu})  $ belongs to the lower (open) half space of the complex plane and, therefore,  $  \mathcal{L} $ is contained in the complement of this cone. This allows us to use   \eqref{resmainnn} and estimate
\begin{align}\label{nooo5}
 \Big \| \frac{1}{ H^{\theta} - r   } \Big \| \leq C \boldsymbol{C}^{n_0 +1} \frac{1}{  {\rm  dist} (r, \mathcal{C}_m(\lambda_1))  },
\end{align}
for every $r \in \mathcal{L}$.  
 Eq.\ \eqref{noes} implies that 
 \begin{align}\label{oc1}
\boldsymbol{C}^{n_0} \leq \boldsymbol{C}  \exp \Big [    -    \log \Big(     \frac{  g^2  \boldsymbol{c} }{   2 \sin(\nu) \rho_0^{1+ \mu/4}  }    \Big )   \Big ]^{  - \frac{\log(\boldsymbol{C})}{ (1 + \mu/4) \log(\rho)  }  }  & =   \boldsymbol{C} \Big (  \frac{  2 \sin(\nu) \rho_0^{1+ \mu/4}   }{ g^2  \boldsymbol{c}  } \Big ) ^{   - \frac{\log(\boldsymbol{C})}{ (1 + \mu/4) \log(\rho)  }   } \notag   \\  & \leq   C  g^{- \iota  }, 
\end{align}
where we use that $\boldsymbol{C}  \rho^{  \frac{1}{2}  \iota   (1 + \mu/4)  } \leq 1  $, see  \eqref{dorm2}.  This together with \eqref{resmainnn}
lead us to  
\begin{align}\label{noooo5}
 \Big \| \frac{1}{ H^{\theta} - r   } \Big \| \leq C g^{- \iota  } \frac{1}{  {\rm  dist} (r, \mathcal{C}_m(\lambda_1))  }, \hspace{1cm} \forall r \in \mathcal{L}.  
\end{align}
It is geometrically clear, because ${ \rm Im   } \lambda_1 < - g^2 \boldsymbol{c} < 0 $ - see \eqref{eq:impartres},  that there is a constant $C$ (that depends on $\boldsymbol{\nu}$ and $m$, but not on $g$)  such that, for  every $r \in \mathcal{L}$, 
\begin{align}\label{nooo1}
|  r - \lambda_1 |  \leq C  {\rm  dist} (r, \mathcal{C}_m(\lambda_1)).
\end{align}
Eqs. \eqref{noooo5} and \eqref{nooo1} yield 
\begin{align}\label{nooooo5}
 \Big \| \frac{1}{ H^{\theta} - r   } \Big \| \leq C g^{-  \iota } \frac{1}{  |  r - \lambda_1 |  }. 
\end{align}
Moreover, we observe from \eqref{projj1} that
\begin{align}\label{repetido}
\norm{\overline{P^\theta_1}  \sigma_1 \Psi^\theta_{\lambda_0}   }
= \norm{(1-P^\theta_1)  \sigma_1 P^\theta_0 \varphi_0\otimes \Omega   }
\leq \norm{(1-P_{\varphi_1}\otimes P_\Omega)   \varphi_1\otimes \Omega   } +Cg
 =  Cg .
\end{align}
Inserting $P^\theta_1+ \overline{P^\theta_1}=1$ in   the left hand side of  \eqref{nooo} and using  \eqref{nooooo5}, we find
\begin{align}\label{nooooo5'}
\norm{ \frac{1}{ H^{\theta} - r   }   \sigma_1 \Psi^\theta_{\lambda_0}   } \leq C\left( \frac{1}{|\lambda_1 -r|}  + g\norm{ \frac{1}{ H^{\theta} - r   }} \right).
\end{align}
This together with $\iota \in (0,1)$ and \eqref{nooooo5} yields 
 that  
\begin{align}
\label{nooooo5''}
\norm{ \frac{1}{ H^{\theta} - r   }   \sigma_1 \Psi^\theta_{\lambda_0}   } \leq C \frac{1}{|\lambda_1 -r|}  .
\end{align}
 From \eqref{nooooo5''} and Remark \ref{RRRR}, we obtain 
 \begin{align}\label{noooooo5}
\int_{\mathcal{L}} \mathrm{d}r\, \norm{ \frac{1}{ H^{\theta} - r   }   \sigma_1 \Psi^\theta_{\lambda_0}   }\leq C   \int_{\mathcal{L}} \frac{1}{  |  r - \lambda_1 |  }  \leq  C    |\log{g}|   .
\end{align}
\end{proof}
\begin{lemma}\label{kkk11}
For every bounded  measurable function $h$, there is a constant $C$ such that  for all natural numbers  $n   \in \N $  and  $\theta\in\mathcal S$
\begin{align}\label{nn0mm}
\int_{   \Big ( 
(B_{0}^{(1)} \cup B_{1}^{(1)}) \cap \mathbb{R}   \Big ) \setminus (\lambda_0-\epsilon_n,\lambda_0+ \epsilon_n)} dz |h(z)|  \Big \| \frac{1}{H^{\theta}-z} \sigma_1 
 \Psi_{\lambda_0}^{\theta}  \Big \|  \leq  
 C   |\log g |  \|  h \|_{\infty},
\end{align}
where $\norm{h}_\infty$ denotes the supremum of $|h|$.
\end{lemma}
\begin{proof}
We set  
\begin{align}
\tilde h(z) = |h(z)|  \Big \| \frac{1}{H^{\theta}-z} \sigma_1 
 \Psi_{\lambda_0}^{\theta}   \Big \| .
\end{align}
  For any natural number $ l\geq 2$,  we set     $ I_l :=B_{0}^{(1)}  \cap [\lambda_0- \epsilon_{l-1},\lambda_0+ \epsilon_{l-1}] 
\setminus  (\lambda_0-\epsilon_{l },\lambda_0+ \epsilon_{l})   $. We define 
$  I_1:=  \mathbb{R}\cap B^{(1)}_0 \setminus  I_2 $. 
Then, we compute 
\begin{align}
\label{split1000}
\int_{( \mathbb{R} \cap B_{0}^{(1)}) \setminus (\lambda_0-\epsilon_n,\lambda_0+ \epsilon_n)} dz \, \tilde h(z)  
 &   =  \sum_{ l = 1   }^{n} \int_{I_l}dz \,\tilde h(z). 
\end{align}
Using Definition \ref{epsilonn}  and  Lemma \ref{kkk1}, we obtain that there is a constant $C$ such that
\begin{align}
\label{noop1}
 \sum_{ l = 1   }^n \int_{I_l}dz \,\tilde h(z)  & \leq C \sum_{ l =2    }^n 
 \boldsymbol{C}^{l+1} \frac{\epsilon_{l-1}}{\rho_l} \|  h\|_{\infty} + C  \|  h\|_{\infty}  \\ \notag & \leq C   \frac{1}{\rho}    \sum_{ l =2    }^n 
 \boldsymbol{C}^{l+1} \rho_{l-1}^{\mu/4}   \|  h\|_{\infty}  + C  \|  h\|_{\infty}
  \leq  C \|  h\|_{\infty},
\end{align}
where we use \eqref{dorm2}. Eq. \eqref{noop1} and Lemma \ref{malhe} imply \eqref{nn0mm}. 
\end{proof}
\begin{lemma}\label{kkk111}
For every bounded  measurable function $h$, there is a constant $C$ such that  for all natural numbers  $n \in \mathbb{N}$   and  $\theta\in\mathcal S$
\begin{align}\label{nn0mmbb}
\int_{ \Gamma_c(\epsilon_n) } dz |h(z)|  \Big \| \frac{1}{H^{\theta}-z} \sigma_1 
 \Psi_{\lambda_0}^{\theta}  \Big \|  \leq   \rho^{\mu/8}_n 
 C  \sup_{z \in \Gamma_c(\epsilon_n)} \{|h(z)| \}. 
\end{align}
\end{lemma}
\begin{proof}
This is a direct consequence of Lemma \ref{kkk1}  and \eqref{dorm2}. 
\end{proof}
\begin{lemma}\label{Lfaltap}
There is a constant $C$ such that (for $s > 0$) 
\begin{align}\label{falta1}
 \Big | \int_{\Gamma_d (R) \cup \Gamma_{-}(\epsilon_n, R) \setminus \Big ( \big ( B_0^{(1)}  \cup B_1^{(1)} \big ) \cap \mathbb{R} \Big ) }   \mathrm{d}z\, e^{-isz}
 \left\langle   \overline{P_1^{\overline \theta}} \sigma_1 \Psi^{\overline \theta}_{\lambda_0},\left( H^\theta-z
 \right)^{-1} \sigma_1 \Psi^\theta_{\lambda_0}\right\rangle \Big | \leq C g \frac{1}{s}. 
\end{align}
\end{lemma}
\begin{proof}
After an integration by  parts, we obtain 
\begin{align}\label{falta2}
\notag  \Big | & \int_{\Gamma_d (R) \cup \Gamma_{-}(\epsilon_n, R) \setminus \Big ( \big ( B_0^{(1)}  \cup B_1^{(1)} \big ) \cap \mathbb{R} \Big ) }   \mathrm{d}z\, e^{-isz}
 \left\langle   \overline{P_1^{\overline \theta}} \sigma_1 \Psi^{\overline \theta}_{\lambda_0},\left( H^\theta-z
 \right)^{-1} \sigma_1 \Psi^\theta_{\lambda_0}\right\rangle \Big | \\  \leq & \frac{1}{s}  \notag 
 \int_{\Gamma_d (R) \cup \Gamma_{-}(\epsilon_n, R) \setminus \Big ( \big ( B_0^{(1)}  \cup B_1^{(1)} \big ) \cap \mathbb{R} \Big ) }   \mathrm{d}z\, \Big |  
 \left\langle   \overline{P_1^{\overline \theta}} \sigma_1 \Psi^{\overline \theta}_{\lambda_0},\left( H^\theta-z
 \right)^{-2} \sigma_1 \Psi^\theta_{\lambda_0}\right\rangle \Big |
 \\ & + \frac{1}{s}  \sum_{z \in \{ e_1+ \delta/2, e_0 - \delta/2  \} } \Big | \left\langle   \overline{P_1^{\overline \theta}} \sigma_1 \Psi^{\overline \theta}_{\lambda_0},\left( H^\theta- z
 \right)^{-1} \sigma_1 \Psi^\theta_{\lambda_0}\right\rangle \Big | .  
\end{align}
Now, we  recall \eqref{repetido}  
\begin{align}\label{buena1}
\norm{ \overline{P_1^{\overline \theta}} \sigma_1 \Psi^{\overline \theta}_{\lambda_0} }
= \norm{ \overline{P_1^{\overline \theta}} \sigma_1 \Psi^{\overline \theta}_{\lambda_0} - 
  \overline{P_{\varphi_1}}\otimes P_{\Omega}  \sigma_1  \varphi_0 \otimes \Omega   }\leq Cg, 
\end{align}
where we use Eq. \eqref{projj1} and  $\sigma_1 \varphi_0 = \varphi_1$. 
Eq \eqref{falta1} is a direct consequence of \eqref{falta2}, \eqref{buena1} and 
\eqref{resA}.
\end{proof}
\begin{lemma}\label{ching1}
For  real numbers $0<q  \leq e^{-1}  <1<Q<\infty$ and $\theta \in \mathcal S$, the term (recall \eqref{eq:G-def})
\begin{align}
\label{eq:thl1t}
R_1(q, Q  )  :  =  &\int^{  Q}_{ q} \mathrm{d}s  \int dr  \,G(r)
 e^{is(r+\lambda_0)}  \\ \notag & \times \int_{\Gamma(\epsilon_n, R)}  \mathrm{d}z\, e^{-isz}
 \left\langle   \overline{P_1^{\overline \theta}} \sigma_1 \Psi^{\overline \theta}_{\lambda_0},\left( H^\theta-z
 \right)^{-1} \sigma_1 \Psi^\theta_{\lambda_0}\right\rangle 
\end{align}
 fulfills
\begin{align}
\label{eq:thl11t}
|R_1(q, Q  ) | \leq & C g\left(  |\log(q)|+   |\log (g)|  \right) .  
\end{align}
Notice that $ R_1(q, Q  ) $ does not depend on $  \epsilon_n $ and $R$, 
 because a change in $\epsilon_n$ and $R$ implies a change in the contour of integration of the analytic function above. 
\end{lemma}
\begin{proof}
First, we    recall  that (see \eqref{buena1}) 
\begin{align}\label{buena}
\norm{ \overline{P_1^{\overline \theta}} \sigma_1 \Psi^{\overline \theta}_{\lambda_0} }
\leq Cg.
\end{align}
A two-fold integration by parts together with the fact $G\in\mathit{C}_c^\infty(\R\setminus\{0\},\C)$ (recall \eqref{eq:G-def})  shows  that there is a constant $C$ such that
\begin{align}
\label{intparts101}
\left| \int dr  \,G(r)
 e^{is(r+\lambda_0)}\right| \leq \frac{C}{1+s^2}, \qquad \forall s\in\R .
\end{align}
 Lemmas \ref{kkk11}, \ref{kkk111} and \ref{Lfaltap}, an Eq. \eqref{buena} imply that
\begin{align}\label{sky0}
|R_1(q, Q  )| \leq C g  \int_q^Q  ds \frac{1}{s^2 +1} \Big (\frac{1}{s}  + 
e^{\epsilon_n Q}  \rho^{\mu/8}_n  
 +  |\log g|    \Big ).
\end{align}
Since $  R_1(q, Q  ) $ does not depend on $n$, we can take $n$  to infinity and obtain the bound 
\begin{align}\label{sky1}
|R_1(q, Q  )| \leq C g  \int_q^Q  ds \frac{1}{s^2 +1} \Big (\frac{1}{s}  +  |\log g|      \Big ).
\end{align}
Eq. \eqref{sky1} and the condition $ 0 < q < e^{-1} $  imply \eqref{eq:thl11t}. Notice that 
 \begin{align}\label{sky2}
 \int_q^Q  ds \frac{1}{(s^2 +1) s } \leq \int_q^{1} \frac{1}{s} + \int_1^{\infty}  \frac{1}{s^2 + 1} \leq C | \log(q)|,  
 \end{align}
 since $  \int_1^{\infty}  \frac{1}{s^2 + 1}  $ is a constant and $ |\log(q)| \geq 1$. 
\end{proof}
\begin{lemma}\label{ching2}
For  real numbers $0<q<1<Q<\infty$ and $\theta \in \mathcal S$,  the term 
\begin{align}
\label{eq:thl2t}
P_1(q, Q  )  =  &\int^{  Q}_{ q} \mathrm{d}s  \int dr  \,G(r)
 e^{is(r+\lambda_0)}  \\ \notag & \times \int_{\Gamma(\epsilon_n, R)}  \mathrm{d}z\, e^{-isz}
 \left\langle  P_1^{\overline \theta} \sigma_1 \Psi^{\overline \theta}_{\lambda_0},\left( H^\theta-z
 \right)^{-1} \sigma_1 \Psi^\theta_{\lambda_0}\right\rangle.
\end{align}
 fulfills
\begin{align}
\label{eq:thl12t}
P_1(q, Q  ) = - 2 \pi  \int \mathrm{d}r     \,G(r)
 \frac{  1  }{  r + \lambda_0 - \lambda_1 } 
 \left\langle  \sigma_1 \Psi^{\overline \theta}_{\lambda_0}, P_1^{ \theta} \sigma_1 \Psi^\theta_{\lambda_0}\right\rangle  + R_2(q, Q),   
\end{align}
where 
\begin{align}
|R_2(q, Q)| \leq C \Big ( q + \frac{1}{Qg^2}\Big )  .
\end{align}
Notice that $ P_1(q, Q  )$ does not depend on $\epsilon_n$ and $R$, because a change in $\epsilon_n$ and $R$ implies a change in the contour of integration of the analytic function above.
\end{lemma}
\begin{proof}
We have that 
\begin{align}
P_1(q, Q  )   =  &\int^{  Q}_{ q} \mathrm{d}s  \int dr  \,G(r)
 e^{is(r+\lambda_0)}   \left\langle  \sigma_1 \Psi^{\overline \theta}_{\lambda_0}, P_1^{ \theta} \sigma_1 \Psi^\theta_{\lambda_0}\right\rangle   \int_{\Gamma(\epsilon_n, R)}  \mathrm{d}z\, \frac{e^{-isz}}{\lambda_1 -z}. 
 \end{align}
 The residue theorem together with methods of complex analysis provides 
 \begin{align}
 \int_{\Gamma(\epsilon_n, R)}  \mathrm{d}z\, \frac{e^{-isz}}{\lambda_1 -z}=2\pi i e^{-is\lambda_1},
 \end{align}
and hence, we  obtain   
 \begin{align}
 \label{finp1}
&P_1(q, Q  )  =       2 \pi i \int^{  Q}_{ q} \mathrm{d}s  \int dr  \,G(r)
 e^{is(r+\lambda_0 - \lambda_1)} 
 \left\langle  \sigma_1 \Psi^{\overline \theta}_{\lambda_0}, P_1^{ \theta} \sigma_1 \Psi^\theta_{\lambda_0}\right\rangle  
\notag \\  
&=   
 2 \pi  \int \mathrm{d}r     \,G(r)
( e^{iQ( r+\lambda_0 - \lambda_1)} - e^{iq(r+\lambda_0 - \lambda_1)}   ) \frac{   1 }{  r + \lambda_0 - \lambda_1 } 
 \left\langle  \sigma_1 \Psi^{\overline \theta}_{\lambda_0}, P_1^{ \theta} \sigma_1 \Psi^\theta_{\lambda_0}\right\rangle 
\notag \\
&  = -2 \pi  \int \mathrm{d}r     \,G(r)
 \frac{  1  }{  r + \lambda_0 - \lambda_1 } 
 \left\langle  \sigma_1 \Psi^{\overline \theta}_{\lambda_0}, P_1^{ \theta} \sigma_1 \Psi^\theta_{\lambda_0}\right\rangle +r_1(Q)+ r_2(q) , 
\end{align}
where 
\begin{align}\label{PF1}
r_1(Q):=2 \pi  \int \mathrm{d}r     \,G(r)
 e^{iQ( r+\lambda_0 - \lambda_1)} \frac{   1 }{  r + \lambda_0 - \lambda_1 } 
 \left\langle  \sigma_1 \Psi^{\overline \theta}_{\lambda_0}, P_1^{ \theta} \sigma_1 \Psi^\theta_{\lambda_0}\right\rangle 
\end{align}
and 
\begin{align}\label{PF2}
r_2(q) := 2 \pi  \int \mathrm{d}r     \,G(r)
( 1 - e^{iq(r+\lambda_0 - \lambda_1)}   ) \frac{   1 }{  r + \lambda_0 - \lambda_1 } 
 \left\langle  \sigma_1 \Psi^{\overline \theta}_{\lambda_0}, P_1^{ \theta} \sigma_1 \Psi^\theta_{\lambda_0}\right\rangle .
\end{align}
 It follows from
\begin{align}
\label{contargq}
| 1  - e^{iq(r+\lambda_0 - \lambda_1)}| \leq | q(r+\lambda_0 - \lambda_1)   |,
\end{align}
 that there is a constant $C$ such that $|r_2(q) | \leq C q$. Applying the integration by parts formula in Eq.\ \eqref{PF1}, we obtain a factor $\frac{1}{Q}$  and the derivative of $  G(r)   \frac{1}{ (r + \lambda_0 - \lambda_1) } $. 
 We obtain 
 \begin{align}\label{PF111}
|r_1(Q)| & \leq C\frac{1}{Q}   \int \mathrm{d}r   \Big (  \,|G(r)|
  \frac{   1 }{ | r + \lambda_0 - \lambda_1 |^2 } +    \,|\frac{d}{dr}G(r)|
  \frac{   1 }{ | r + \lambda_0 - \lambda_1 | }  \Big )
| \left\langle  \sigma_1 \Psi^{\overline \theta}_{\lambda_0}, P_1^{ \theta} \sigma_1 \Psi^\theta_{\lambda_0}\right\rangle | 
\notag \\
 & \leq C \frac{1}{Q }( \frac{1}{g^2}  + |\log(g)| ) \leq  C 
\frac{1}{Q g^2 } ,
\end{align}
 where we use \eqref{rodos1}, with $x = 1$ and $x= 0$, and $r + \lambda_0$ instead of $r$. 
\end{proof}
\begin{lemma}\label{ching2tt}
For  real numbers $0<q<1<Q<\infty$ and $\theta \in \mathcal S$, we define the term 
\begin{align}
\label{eq:thl2ttt}
\widetilde P_1(q, Q  )  : =  &\int^{  Q}_{ q} \mathrm{d}s  \int dr  \,G(r)
 e^{is(r -  \lambda_0)}  \\ \notag & \times \int_{ \widetilde \Gamma(\epsilon_n, R)}  \mathrm{d}z\, e^{ isz}
 \left\langle  P_1^{ \theta} \sigma_1 \Psi^{ \theta}_{\lambda_0},\left( H^{\overline \theta}-z
 \right)^{-1} \sigma_1 \Psi^{\overline \theta}_{\lambda_0}\right\rangle,
\end{align}
where $  \widetilde \Gamma(\epsilon_n, R) $ is a positively oriented  curve  obtained by conjugating the elements of $  \Gamma(\epsilon_n, R) $.
It follows that
\begin{align}
\label{eq:thl12ttt}
\widetilde P_1(q, Q  ) =   2 \pi  \int \mathrm{d}r     \,G(r)
 \frac{  1  }{  r -  \lambda_0 + \overline{\lambda_1} } 
 \left\langle  P_1^{ \theta} \sigma_1 \Psi^\theta_{\lambda_0},  \sigma_1 \Psi^{\overline \theta}_{\lambda_0}\right\rangle   +  \widetilde R_2(q, Q),   
\end{align}
where
\begin{align}
|\widetilde R_2(q, Q)| \leq C   \Big ( q + \frac{1}{Qg^2}\Big ) .
\end{align}
\end{lemma}
\begin{proof}
This proof is very similar to the proof of Lemma \ref{ching2}: 
We have that 
\begin{align}
\widetilde P_1(q, Q  )   =  & - 2 \pi i \int^{  Q}_{ q} \mathrm{d}s  \int dr  \,G(r)
 e^{is(r - \lambda_0 +  \overline{\lambda_1})}  \left\langle  P_1^{ \theta} \sigma_1 \Psi^\theta_{\lambda_0},  \sigma_1 \Psi^{\overline \theta}_{\lambda_0}\right\rangle 
 \end{align}
and hence, we infer 
 \begin{align}
 \label{finp1tt}
P_1(q, Q  )  = &  
  -2 \pi  \int \mathrm{d}r     \,G(r)
( e^{iQ( r -\lambda_0 + \overline{\lambda_1})} - e^{iq(r - \lambda_0 +
 \overline{\lambda_1})}   ) \frac{   1 }{  r - \lambda_0 + \overline{\lambda_1} } 
\left\langle  P_1^{ \theta} \sigma_1 \Psi^\theta_{\lambda_0},  \sigma_1 \Psi^{\overline \theta}_{\lambda_0}\right\rangle 
\notag \\
    = &  2 \pi  \int \mathrm{d}r     \,G(r)
 \frac{  1  }{  r - \lambda_0 + \overline{\lambda_1} } 
\left\langle  P_1^{ \theta} \sigma_1 \Psi^\theta_{\lambda_0},  \sigma_1 \Psi^{\overline \theta}_{\lambda_0}\right\rangle  + \widetilde R_2(q,Q).  
\end{align}
We conclude the proof as in the proof of Lemma \ref{ching2}. 
\end{proof}

\subsection{Proof of Theorem \ref{FKcor}}\label{main}

In this section, we give the proof of the main theorem based on the previous
results.
\begin{proof}[Proof of Theorem \ref{FKcor}]
Let $h,l\in\mathfrak h_0$; c.f.\ \eqref{def:h0}.
    Recall  the definition of $W$ given in \eqref{def:W1st} and the form factor
$f$ in \eqref{eq:f}.
Since $f\in \mathit C^\infty(\R^3\setminus \{ 0\},\C)$ 
we find
\begin{align}
\label{eq:hflftt}
hf,lf, W\in\mathfrak{h}_0.
\end{align}
Theorem~\ref{intker}, i.e., equation~\eqref{eq:Tprecise1}, together with
Lemma~\ref{lemmadiff} (iv) yields
\begin{align}
T(h, l)
&=
-2\pi i g\norm{\Psi_{\lambda_0}}^{-2} \left\langle a_-(W) \sigma_1 \Psi_{\lambda_0}, \Psi_{\lambda_0}\right\rangle
=-2\pi i g\norm{\Psi_{\lambda_0}}^{-2} \left\langle [a_-(W) ,\sigma_1] \Psi_{\lambda_0}, \Psi_{\lambda_0}\right\rangle ,
\end{align}
and
furthermore, recalling $\omega(k)=|k|$,
 equation
\eqref{a_-} in Lemma~\ref{lemmadiff}  (ii)  
implies
\begin{align}
\label{eq:thl0tt}
T(h,l)
&=
 2\pi(ig)^2\norm{\Psi_{\lambda_0}}^{-2} \int_{-\infty}^0 \mathrm{d}s   \,
\overline{\langle W_s,f\rangle_2}
\left\langle \left[e^{isH}\sigma_1
e^{-isH},  \sigma_1\right]  \Psi_{\lambda_0}, \Psi_{\lambda_0}\right\rangle 
\notag \\
&=   2\pi  g^2\norm{\Psi_{\lambda_0}}^{-2} \int^{\infty}_0 \mathrm{d}s \,
    \langle f, W_{-s}\rangle_2
    \left\langle \left[e^{-isH}\sigma_1
e^{isH},  \sigma_1\right]  \Psi_{\lambda_0}, \Psi_{\lambda_0}\right\rangle 
\notag \\
&= 2 \pi  g^2 \norm{\Psi_{\lambda_0}}^{-2} \left(  T^{(1)}- T^{(2)} \right) 
,
\end{align}
where  we used the abbreviations 
\begin{align}
\label{eq:abr.termstt}
T^{(j)}:=    T^{(j)}(0, \infty) 
\end{align}
 for $j=1,2$ with
\begin{align}
\label{eq:first-termtt}
T^{(1)}(q, Q)  :=  &\int_q^Q  \mathrm{d}s  \int \mathrm{d^3} k   \,
W(k)f(k) e^{is(|k|+\lambda_0)}  \left\langle   \sigma_1 \Psi_{\lambda_0},
e^{-isH} \sigma_1\Psi_{\lambda_0}\right\rangle 
\\ \notag = & \int_q^Q  \mathrm{d}s  \int \mathrm{d} r  G(r) e^{is(r+\lambda_0)}  \left\langle   \sigma_1 \Psi_{\lambda_0},
e^{-isH} \sigma_1\Psi_{\lambda_0}\right\rangle 
\end{align}
and 
\begin{align}
T^{(2)}(q,Q) :=   \int_q^Q  \mathrm{d}s  \int \mathrm{d} r   G(r) e^{is(r-\lambda_0)}     \left\langle \sigma_1  \Psi_{\lambda_0},
e^{isH}\sigma_1\Psi_{\lambda_0}\right\rangle .
\label{eq:second-termtt}
\end{align}
We recall  the definitions \eqref{def:W1st} and \eqref{eq:G-def}:
\begin{align}
\label{def:W1sttt}
W( k) =|k|^2 l(k) \int\mathrm{d}\Sigma \, \overline{h(|k|,\Sigma)}f(|k|,\Sigma), 
\hspace{.5cm} G(r)=   \int \mathrm{d}\Sigma \mathrm{d}\Sigma' \,  r^4  \overline{h(r,\Sigma)} l(r,\Sigma') f(r)^2.
\end{align}
We observe that there is a constant $C$ such that 
\begin{align}\label{comien}
| T^{(i)}(q, Q) - T^{(i)}(0, Q)  | \leq C q.  
\end{align}
 We start with analyzing the term  $  T^{(1)}(q, Q) $.
Lemma \ref{laplace}  together with the identity $P_1^{\overline \theta}+\overline{P_1^{\overline \theta}}=1$ allows us to write this term as
\begin{align}\label{comien2'}
 T^{(1)}(q, Q) = & \frac{1}{2 \pi i } P_1(q, Q) +
 \frac{1}{2 \pi i } R_1(q, Q) 
\end{align}
for all $0<q<Q<\infty$.
 Here, $P_1(q, Q)$ and $R_1(q, Q) $ are defined in   \eqref{eq:thl2t} and  \eqref{eq:thl1t}, respectively. Moreover, Lemma \ref{ching2} implies
\begin{align}\label{comien2}
 T^{(1)}(q, Q) &= - 2 \pi  \frac{1}{2 \pi i}  \int \mathrm{d}r     \,G(r)
 \frac{  1  }{  r + \lambda_0 - \lambda_1 } 
 \left\langle  \sigma_1 \Psi^{\overline \theta}_{\lambda_0}, P_1^{ \theta} \sigma_1 \Psi^\theta_{\lambda_0}\right\rangle 
 \notag \\
 &+ \frac{1}{2 \pi i} \Big (   R_1(q, Q) +    R_2(q, Q) \Big ), 
\end{align}
where Lemmas \ref{ching1} and  \ref{ching2} provide the estimates: 
\begin{align}\label{comien3}
| R_1(q, Q)| \leq   C g\left(  |\log(q)|+   |\log (g)|  \right)  , \qquad
 |R_2(q, Q)| \leq C 
  (q + \frac{1}{g^2 Q})  
\end{align}
  for $0<q \leq e^{-1}  < 1 < Q$. 
 As explained in Lemmas \ref{ching1} and \ref{ching2}, 
the terms $  P_1(q, Q) $ and  $ R_1(q, Q)  $ do not depend on $ n $ and $R$ because both are given by  contour  integrals of analytic functions and a change of these parameters signifies a change in the contour of integration only. Taking the limit $Q$ to infinity
and $q = g$, we obtain from Eqs. \eqref{comien}, \eqref{comien2} and \eqref{comien3}:
\begin{align}\label{comien4}
T^{(1)}(0, \infty) =  i   \int \mathrm{d}r     \,G(r)
 \frac{  1  }{  r + \lambda_0 - \lambda_1 } 
 \left\langle  \sigma_1 \Psi^{\overline \theta}_{\lambda_0}, P_1^{ \theta} \sigma_1 \Psi^\theta_{\lambda_0}\right\rangle + R_3
\end{align}
and that there is a constant $C$ such that
\begin{align}\label{comien6}
|R_3| \leq C  |\log (g)|    .  
\end{align}
The term  $T^{(2)}(0, \infty)$ can  be inferred 
by repeating the calculation with $\theta$ replaced by $\overline\theta$ and reflecting the
path of integration $\Gamma (\epsilon,R)$ on the real axis when applying Lemma
\ref{laplace}. In this case one  has to consider  the Hamiltonian
$H^{\overline \theta}$.    Notice that in this case the factor $\frac{1}{2 \pi i}$ in Eq.\ \eqref{eq:laplace}  is substituted by $- \frac{1}{2 \pi i}$,    which is produced from  the change of orientation of the integration curve.    
Due to the similarity of the calculation, we omit  the  proof
   and   only state the
result (it follows from Lemma \ref{ching2tt} and similar  computations)
\begin{align}\label{comien7}
T^{(2)}(0, \infty) =\frac{- 1}{2 \pi i }2 \pi  \int \mathrm{d}r     \,G(r)
 \frac{  1  }{  r -  \lambda_0 + \overline{\lambda_1} }  \left\langle  \sigma_1 \Psi^{ \theta}_{\lambda_0}, P_1^{ \overline\theta} \sigma_1 
 \Psi^{\overline \theta}_{\lambda_0}\right\rangle 
   + R_4
\end{align}
and that there is a constant $C$ such that
\begin{align}\label{comien8}
|R_4| \leq C   |\log (g)|   .  
\end{align}
The identities \eqref{eq:thl0tt},   \eqref{comien4}  and \eqref{comien7}, together with  \eqref{comien6}, \eqref{comien8} and \eqref{projj1}   imply
\begin{align}\label{mm}
T(h,l) & =  2 \pi  g^2 \norm{\Psi_{\lambda_0}}^{-2} \left(  T^{(1)}- T^{(2)} \right) + R
 \\ \notag & =   2 \pi i g^2 \norm{\Psi_{\lambda_0}}^{-2}   \int \mathrm{d}r     \,G(r)
\Big (  \frac{  1  }{  r + \lambda_0 - \lambda_1 } -  \frac{  1  }{  r -  \lambda_0 + \overline{\lambda_1} } \Big ) + R
\\ \notag & =   4\pi i g^2  \norm{\Psi_{\lambda_0}}^{-2}  \int \mathrm{d}r     \,G(r)
\Big (  \frac{ \Re \lambda_1 -  \lambda_0 }{ \big (  r + \lambda_0 - \lambda_1 \big ) \big (  r -  \lambda_0 + \overline{\lambda_1}    \big )  }  \Big ) + R, 
\end{align} 
where 
$|R| \leq C  |\log (g)|  $. 
\end{proof}
\begin{remark}\label{Constant}
The constant $C(h,l)$ in Theorem \ref{FKcor} depends on $h$ and $l$. From our methods, this dependence can be made explicit. However, for the sake of simplicity and clarity we do not present this analysis in this paper, but indicate instead how to do it. The key ingredients are Eqs. \eqref{intparts101} and \eqref{PF2} (notice that \eqref{PF1} does not play a role because the corresponding term vanishes    when   $Q$ tends to infinity). These terms give a contribution of the form
\begin{align}\label{contri1}
C\int dr \Big [ |G(r)| + \Big |\frac{d}{d r }G(r)\Big | +  \Big | \frac{d^2}{dr^2}G(r) \Big | \Big ],
\end{align} 
for a constant $C$ that does not depend on $h$ and $l$. Moreover, with respect to \eqref{intparts101}, a minor change in the proof of  Lemma \ref{ching1} would make the second derivative term unnecessary because we have an extra factor of the form $s^{-1}$ in  \eqref{falta1}.    This is essentially the only necessary contribution that comes from $h$ and $l$. However, in order to simplify our  final formula, we substituted the inner products in  Eqs.\ 
 \eqref{comien4} and \eqref{comien7} by the constant $1$ (using  \eqref{projj1}). This produces (explicit) error  terms  that contribute differently as \eqref{contri1}, as we can see from our arguments below  \eqref{comien8}. 
\end{remark}

\begin{appendix}
\section{Standard Estimates}
\label{app:sa}
In the following we shall use the well-known standard inequalities
\begin{align} 
\begin{split}
    \|a(h)\Psi\|&\leq \|h/\sqrt\omega\|_2 \, \|H_f^{1/2}\Psi\|
    \\
     \|a(h)^*\Psi\|&\leq \|h/\sqrt\omega\|_2 \, \|H_f^{1/2}\Psi\| + \|h \|_2 \, \| \Psi\| 
    \end{split}
      \label{eq:st-est}
\end{align}
which hold for all  $h , h/\sqrt\omega \in\mathfrak h$  and $\Psi\in\mathcal H$ such that the
left- and right-hand side are well-defined; see \cite[Eq.\ (13.70)]{spohn_dynamics_2008}. 
\begin{lemma}
\label{lemma:standardest1}
Let $h, h/\sqrt{\omega}\in \mathfrak{h}$. Then, we have the following 
estimates:
\begin{align}
\label{eq:standartesta}
\norm{a(h)^* (H_f+1)^{-\frac{1}{2}}}
&\leq 
\norm{h}_2 + \norm{h/\sqrt{\omega}}_2
,\\
\norm{a(h) (H_f+1)^{-\frac{1}{2}}}&\leq
\norm{h/\sqrt{\omega}}_2 ,\\
\norm{ V \left( H_f +1  \right)^{-\frac{1}{2}}}
&\leq \norm{f}_2 +2\norm{f/\sqrt{\omega}}_2 .
\label{eq:Vest}
\end{align}
\end{lemma}
\begin{proof}
Let $\Psi\in \mathcal F [\mathfrak{h}]$ with $\|\Psi \|_{\mathcal H}=1$.
Applying \eqref{eq:st-est} and the spectral theorem, we find
\begin{align}
    \| a(h)^* (H_f+1)^{-\frac{1}{2}}\Psi \|
& \leq \|h\|_2  \|(H_f+1)^{ -\frac12}\Psi\|+ \|h/\sqrt\omega\|_2 \|H_f^{\frac12}(H_f+1)^{ -\frac12}\Psi\|
\notag \\
    &\leq \|h\|_2+\|h/\sqrt\omega\|_2,
    \\
    \| a(h) (H_f+1)^{-\frac{1}{2}}\Psi \|
    &
    \leq \|h/\sqrt\omega\|_2 \|H_f^{\frac12}(H_f+1)^{  - \frac12}\Psi\| 
    \leq \|h/\sqrt\omega\|_2.
\end{align} 
The inequality \eqref{eq:Vest} is implied by the boundedness of $\sigma_1$
and the triangle inequality:
\begin{align}
&\norm{ V \left( H_f +1  \right)^{-\frac{1}{2}}} 
 \leq \norm{ \sigma_1  \otimes a(f)\left( H_f +1  \right)^{-\frac{1}{2}}  }+ \norm{\sigma_1  \otimes a(f)^*\left( H_f +1  \right)^{-\frac{1}{2}}  }
\notag \\
& \leq  \norm{ a(f) \left( H_f +1  \right)^{-\frac{1}{2}} } +\norm{ a(f)^*\left( H_f +1  \right)^{-\frac{1}{2}}  }
 \leq \norm{f}_2 +2\norm{f/\sqrt{\omega}}_2.
\end{align}
This completes the proof.
\end{proof}
As preparation of the proof of Lemma~\ref{lemmadiff}  (in Appendix \ref{app:welldef} below)   we recall that the
Hamiltonians $H$, c.f.\ \eqref{eq:H}, as well as $H_f$, c.f.\ \eqref{h0def},
are self-adjoint on the common domain $D(H)=\mathcal K\otimes \mathcal D(H_f)$  and bounded below by the constant $b\in\R$;
c.f.\ Proposition~\ref{thm:Hsa}  and \eqref{bbdbelow}. By spectral calculus we can define the operators
$H_f^{1/2}$,  $(H-b+1)^{1/2}$  and $(H_f+1)^{-1/2}$,  $(H-b+1)^{-1/2}$  which are
closed and densely defined and the latter two are even bounded by $1$.  For the
proof Lemma~\ref{lemmadiff} we shall need the following lemma.

\begin{lemma}
    \label{lem:bhh}
    The following operators are bounded:
    \begin{align}
        H_f^{\frac12}(H-b+1)^{-\frac12} \label{eq:bhh1},\\
        (H-b+1)^{\frac12}(H_f+1)^{-\frac12}. \label{eq:bhh2}
    \end{align}
\end{lemma}
\begin{proof}
    Let $\Psi\in\mathcal
    H$ with $\|\Psi\|=1$. The boundedness of \eqref{eq:bhh1} follows from the
    equality
    \begin{align}
        \|H_f^{\frac12}(H-b+1)^{-\frac12}\Psi\|^2 
        &= \langle(H-b+1)^{-\frac12}\Psi,H_f(H-b+1)^{-\frac12}\Psi\rangle
    \notag    \\ 
        &= \langle(H-b+1)^{-\frac12}\Psi,(H-K-gV)(H-b+1)^{-\frac12}\Psi\rangle
    \end{align}
    and the fact that $K$ is bounded by $|e_1|$ and that for all $\epsilon>0$ 
    \begin{align}
       & |\langle(H-b+1)^{-\frac12}\Psi,gV(H-b+1)^{-\frac12}\Psi\rangle|
        \leq
        \|(H-b+1)^{-\frac12}\Psi\|\,\|gV(H-b+1)^{-\frac12}\Psi\|
    \notag    \\
        &\leq \frac{g}{\epsilon}  2\|f/\sqrt\omega\|_2 \, 
        \epsilon \|H_f^{\frac12}(H-b+1)^{-\frac12}\Psi\|  
        +\|f\|_2  \|  \Psi  \|  
   \notag     \\
        &\leq \left(\frac{g}{\epsilon} 2\|f/\sqrt\omega\|_2\right)^2
        + \epsilon^2
        \|H_f^{\frac12}(H-b+1)^{-\frac12}\Psi\|^2  
        +\|f\|_2  
    \end{align} 
    holds, which is a consequence of \eqref{eq:st-est}.
    Choosing $0<\epsilon<1$ an explicit bound is
    \begin{align}
        \|H_f^{\frac12}(H-b+1)^{-\frac12}\Psi\|^2 
        \leq
        \frac{1 + |e_1| + \left(\frac{g}{\epsilon} 2\|f/\sqrt\omega\|_2\right)^2    
        +\|f\|_2    }{1-\epsilon^2}
        <\infty.
    \end{align}
    The boundedness of \eqref{eq:bhh2} is implied by
    \begin{align}
        \|(H-b+1)^{\frac12}(H_f+1)^{-\frac12}\Psi\|^2 
        &=
        \langle(H_f+1)^{-\frac12}\Psi,(K+H_f+gV-b+1)(H_f+1)^{-\frac12}\Psi\rangle
    \end{align}
    and, again as a consequence of \eqref{eq:st-est}, 
    \begin{align}
        |\langle(H_f+1)^{-\frac12}\Psi,gV(H_f+1)^{-\frac12}\Psi\rangle|
        &\leq g  2\|f/\sqrt\omega\|_2 \, 
        \|H_f^{\frac12}(H_f+1)^{-\frac12}\Psi\|    
        +\|f\|_2    \\ \notag   &\leq  \|f\|_2+ 2 \|f/\sqrt\omega\|_2.
    \end{align} 
\end{proof}

\section{Proofs for Section \ref{sec:dil}}
\label{app:dil}
 It is well-known that there is a dense domain of analytic vectors; for example
$$\mathcal D = \left\{ \chi_{[-R,R]}(A)\Psi : \Psi \in \mathcal H , R>0
\right\}$$ with  $A$ being the generator of $U_\theta$ and $\chi$ the
corresponding spectral projection (c.f.\ \cite{bach,jaksic}).

\begin{proof}[Proof of Lemma \ref{lemma:specdilh0}]
    Let $\theta\in\C$. Definition in \eqref{h0def} implies that $H^\theta_0=K
    \otimes \mathbbm
    1_{\mathcal F[\mathfrak h]} + \mathbbm 1_{\mathcal K} \otimes H^\theta_f $
    is a sum of commuting self-adjoint operators and
    $\sigma(K)=\left\lbrace e_0,e_1 \right\rbrace$. As shown in
    \cite{reedsimon2}, we have $\sigma(H_f)=\R_0^+$ and it
    follows from the definition of $H_f^\theta=e^{-\theta}H_f$ in
    \eqref{eq:Hftheta} that $\sigma (H^\theta_f)= \left\lbrace  e^{-\theta} r :
    r\geq 0 \right\rbrace$. The claim then follows from  the  spectral  theorem for
    two commuting normal operators.
\end{proof}

\section{Asymptotic creation/annihilation operators}
\label{app:welldef}

\begin{proof}[Proof of Lemma~\ref{lemmadiff}]
    Let $h,l\in\mathfrak h_0$ and $\Psi\in \mathcal K \otimes \mathcal D(H_f^{1/2})$. Thanks
    to Lemma~\ref{lem:bhh} we have $\mathcal K \otimes\mathcal D(H_f^{\frac12})=\mathcal
    D((H - b + 1 )^{\frac12})$. We prove claims (i)-(vi) separately:
    \begin{enumerate}
        \item[(ii)] The subspace of $\mathcal
            H_0$, defined in \eqref{eq:H0},   is dense  in the domain of $(H-b+1)^{1/2}$  w.r.t.\ the graph norm
            {$\|\cdot\|_{(H-b+1)^{1/2}}$} of $(H-b+1)^{\frac12}$ so that
            there is a sequence $(\Psi_n)_{n\in\N}$ in $\mathcal K \otimes\mathcal
            F_{\text{fin}}[\mathfrak h_0]$ with $\Psi_n\to\Psi$ in this norm as
            $n\to\infty$.  For all $n\in\N$, the definition in \eqref{asymptop}
            together with the group properties $(e^{-itH})_{t\in\R}$, in
            particularly, the strong continuous differentiability on $D(H)$,
            justify
            \begin{align}
                a_t(h)\Psi_n 
                &= e^{itH}a(h_t)e^{-itH}
             = a(h)\Psi_n + \int_0^t ds\, \frac{d}{ds} e^{isH} a(h_s)
                e^{-isH}\Psi_n
                \notag
                \\
                &= a(h)\Psi_n -ig \int_0^t ds\, \langle h_s,f\rangle_2 
                e^{isH} \sigma_1 e^{-isH}\Psi_n,
                \label{eq:a-t-psi-n}
            \end{align}
            where the last integrand was computed by observing the 
            CCR (c.f.\ \eqref{eq:ccr})
            \begin{align}
                [V ,a(h_s)]
                &= \sigma_1 \otimes
                [a(f)+a(f)^*,a(h_s)]=-\sigma_1  \left\langle h_s,
                f\right\rangle_2.
            \end{align}
            We may now take the limit $n\to\infty$ of
            identity \eqref{eq:a-t-psi-n} and find
            \begin{align}
                a_t(h)\Psi
                &= a(h)\Psi -ig \int_0^t ds\, \langle h_s,f\rangle_2 \,
                e^{isH} \sigma_1 e^{-isH}\Psi
                \label{eq:ident-t}
            \end{align}
            because of the following two ingredients:
            First, by definition \eqref{asymptop}, the standard estimate
            \eqref{eq:st-est} and Lemma~\ref{lem:bhh}, for all $m\in\mathfrak
            h_0$, there is a finite
            constant $C_{\eqref{eq:m-t-const}}$ such that
            \begin{align}
                \|a_t(m)(\Psi-\Psi_n)\| 
                &= \|a(m_t)(H-b+1)^{-\frac12}e^{-itH}(H-b+1)^{\frac12}(\Psi-\Psi_n)\|
                \notag
                \\
                &\leq 
                \|m/ \sqrt \omega\|_2 
                \, \|H_f^{\frac12}(H-b+1)^{-\frac12}\| 
                \, \|(H-b+1)^{\frac12}(\Psi-\Psi_n)\|
                \notag
                \\
                &= C_{\eqref{eq:m-t-const}}
                \|\Psi-\Psi_n\|_{(H-b+1)^{1/2}},
                \label{eq:m-t-const}
            \end{align}
            and likewise
            \begin{align}
                \|a(m)(\Psi-\Psi_n)\| 
                &= \|a(m)(H-b+1)^{-\frac12}(H-b+1)^{\frac12}(\Psi-\Psi_n)\|
                \notag
                \\
                &\leq 
                \|m/ \sqrt \omega\|_2 
                \, \|H_f^{\frac12}(H-b+1)^{-\frac12}\| 
                \, \|(H-b+1)^{\frac12}(\Psi-\Psi_n)\|
                \notag
                \\
                &= C_{\eqref{eq:m-t-const}}
                \|\Psi-\Psi_n\|_{(H-b+1)^{1/2}}.
            \end{align}
            Second, the integrand in \eqref{eq:a-t-psi-n} is continuous in $s$ 
            and, for sufficiently large $n$, fulfills an $n$-independent bound 
            \begin{align}
                \|e^{isH} \sigma_1 e^{-isH}(\Psi-\Psi_n)\|
                \leq
                \|\sigma_1 \|\,\|\Psi-\Psi_n\|
                \leq 1
            \end{align}
            so dominated convergence can be applied to interchanging
            the integral and the $n\to\infty$ limit to prove \eqref{eq:ident-t}.

            Finally, a stationary phase argument in $\omega(k)=|k|$ as well as
            the facts that $h\in\mathfrak h_0$ and $f\in\mathcal
            C^\infty(\R\setminus\{0\})$, c.f.\ \eqref{eq:f}, provide the estimate 
            \begin{align}
                \langle h_s,f\rangle = C \frac{1}{1 + |s|^2}
                \label{eq:stat-phase}
            \end{align}
            for  all $s\in\R$, thanks to a two-fold partial integration.
            Hence, me way finally carry out the limit $t\to\pm\infty$ to find
            \begin{align}
                a_\pm(h)\Psi 
                = 
                \lim_{t\to\pm\infty} a_t(h)\Psi =
                a(h)\Psi -ig 
                \int_0^{\pm\infty} 
                ds\, \langle h_s,f\rangle_2 
                e^{isH} \sigma_1 e^{-isH}\Psi
            \end{align}
            as the indefinite integral exists thanks to \eqref{eq:stat-phase}
            and the continuity of the integrand in $s$. We omit the proof for
            the asymptotic creation operator $a^*_\pm $ as the argument is almost
            the same.
        \item[(i)] This follows from (ii).
        \item[(iii)] Next, we calculate
            \begin{align}
                e^{-isH} a_- (h)^*\psi
                &=\lim\limits_{t\to -\infty} e^{-isH}  e^{itH}a(h_t)^*e^{-itH} \psi
                \notag \\
                &=\lim\limits_{t\to -\infty}  e^{i(t-s)H}a(h_{(t-s)+s})^* e^{-i(t-s)H} e^{-isH} \psi
                \notag \\
                &=\lim\limits_{t'\to -\infty}  e^{it'H} a(h_{t'+s})^* e^{-it'H} e^{-isH} \psi
                = a_- (h_s)^*  e^{-isH}\psi 
            \end{align}
            which proves the pull-through formula in (iii).
        \item[(iv)] First, for all $t\in\R$ we observe
            \begin{align}
                \|a_t(h)\Psi_{\lambda_0}\|
                =
                \|e^{itH}a(h_t)e^{-itH}\Psi_{\lambda_0}\|
                =
                \|a(h_t)\Psi_{\lambda_0}\|
                \label{eq:a-psi-lambda}
            \end{align}
            due to the ground state property  in \eqref{gsprop}.  
            Second, for
            $\Psi=\Psi_{\lambda_0}\in \mathcal D(H)\subset \mathcal K \otimes \mathcal
            D(H_f^{1/2})$, we employ the
            same sequence $(\Psi_n)_{n\in\N}$ as in (ii) to compute 
            \begin{align}
                \label{eq:aht-psi-n}
                \|a(h_t)\Psi_n\|^2=\sum_{l\in\N}\sqrt{l+1}\int d^3k_1\ldots d^3k_l
                \left|
                \int d^3k \, e^{it\omega(k)}\overline{h(k)}\psi_n^{(l+1)}(k,k_1,\ldots,k_l)
                \right|^2,
            \end{align}
            where we used the Fock vector representation
            $\Psi_n=(\psi_n^{(l)})_{l\in\N_0}$. We observe that $\Psi_n\in\mathcal
            H_0$ implies $\psi^{(l)}_n\in \mathcal K \otimes C_0^{\infty} ( \mathbb{R}^{3l} \setminus \{ 0\} )$ and, by definition of
            $\mathcal H_0$, c.f.\ \eqref{eq:H0},  there is a constant
            $L$ such that 
            $\psi^{(l)}_{n}=0$ for  $l\geq L$ . A stationary phase argument in
            $\omega(k)=|k|$ and a partial integration in $k$ gives
            \begin{align}
              &  \left|\int d^3k \, e^{it\omega(k)}\overline{ h(k)}  \psi_n^{(l+1)}(k,k_1,\ldots,k_{l}) \right|
               \notag \\ & \leq \frac{1}{t}
                \int d^3k \,  |k|^{-2}  | \partial_{|k|} (    |k|^2  \overline{h( |k|,\Sigma)}\psi_n^{(l+1)}( |k|,\Sigma, |k_1|,\Sigma_1,\ldots, |k_l|,\Sigma_l)) | ,
            \end{align}
        where we use spherical coordinates $k=(|k|,\Sigma)$   and $k_i=(|k_i|,\Sigma_i)$. Here, $\Sigma$ and $\Sigma_i$ denote the solid angles.  Then, we find
            \begin{align}
                \eqref{eq:aht-psi-n}\leq & \frac{1}{t} 
                \sum_{0\leq l< L}  \sqrt{l+1} \int d^3k_1\ldots d^3k_l
                \\\notag
               & \times \left(\int d^3k \, |k|^{-2}  | \partial_{|k|} ( |k|^{2}  \overline{h(|k|,\Sigma)} \Psi_n^{(l+1)}( |k|,\Sigma, |k_1|,\Sigma_1,\ldots, |k_l|,\Sigma_l) |\right)^2
            \end{align} 
            which converges to zero for $t\to\pm\infty$.
            In conclusion, for all $n\in\R$ we have
            \begin{align}
               \lim_{t\to\pm\infty}a(h_t)\Psi_n = 0. 
               \label{eq:a-t-W-0}
            \end{align}
            Moreover, there is a $t$-independent, finite constant
            $C_\eqref{eq:bound-H+i}(h)$ such
            that
            \begin{align}
                \|a_t(h)(\Psi_{\lambda_0}-\Psi_n)\|
                &= 
                \|e^{itH}a(h_t)e^{-itH}(\Psi_{\lambda_0}-\Psi_n)\|
               \notag  \\
                &= 
                \|a(h_t)(H-b+1)^{-\frac12}e^{-itH}(H-b+1)^{\frac12}(\Psi-\Psi_n)\|
                \notag
                \\
                &\leq 
                \||h|/ \sqrt \omega\|_2 
                \, \|H_f^{\frac12}(H-b+1)^{-\frac12}\| 
                \|\Psi-\Psi_n\|_{(H-b+1)^{1/2}}
              \notag  \\
                &=C_{\eqref{eq:bound-H+i}}(h)
                \|\Psi-\Psi_n\|_{(H-b+1)^{1/2}}
                \label{eq:bound-H+i}
            \end{align}
            and
            \begin{align}
                \| a_\pm(h)\Psi_{\lambda_0} \| &\leq \lim_{t\to\pm\infty} \left(
                    \|a_t(h)(\Psi_{\lambda_0}-\Psi_n)\| 
                    +
                    \|a_t(h)\Psi_n\| 
                \right)
          \notag      \\
                &\leq
                C_{\eqref{eq:bound-H+i}}(h)
                \|\Psi-\Psi_n\|_{(H-b+1)^{1/2}}
                \label{eq:bound-H+i2}
            \end{align}
            holds true for all $n\in\N$, where we have use the standard
            inequalities \eqref{eq:st-est}, Lemma~\ref{lem:bhh} and \eqref{eq:a-t-W-0}. 
            Taking the limit $n\to\infty$ proves the claim (iv).
        \item[(v)] We consider the same sequence $(\Psi_n)_{n\in\N}$ as in
            (iv) and, for all $n\in\N$, we observe that, by (i) and definition
            in \eqref{asymptop}, it holds
            \begin{align}
                \langle a(h)^*_\pm\Psi_{\lambda_0},
                a(l)^*_\pm\Psi_{\lambda_0}\rangle
                &=
                \lim_{t\to\pm\infty}
                \langle a(h_t)^*\Psi_{\lambda_0},
                a(l_t)^*\Psi_{\lambda_0}\rangle.
                \label{eq:comm-lambda}
            \end{align}
            Furthermore, using the CCR in \eqref{eq:ccr}, we find for all
            $n\in\N$ that
            \begin{align}
               & \langle a(h_t)^*\Psi_{\lambda_0},
                a(l_t)^*\Psi_n\rangle
                =
                \langle \Psi_{\lambda_0},
                a(h_t)a(l_t)^*\Psi_n\rangle
                \\ \notag
                & =
                \langle \Psi_{\lambda_0},
                \left(a(l_t)^*a(h_t)+[a(h_t),  a(l_t)^*]\right)\Psi_n\rangle
                =
                \langle a(l_t)\Psi_{\lambda_0},
                a(h_t)\Psi_n\rangle
                + \langle\Psi_{\lambda_0},\Psi_n\rangle \, \langle h,
                l\rangle_2
                \label{eq:comm-n}
            \end{align}
            holds. We may control the limit $n\to\infty$ of this identity by
            \begin{align}
                &|\langle a(h_t)^*\Psi_{\lambda_0},
                a(l_t)^*(\Psi_{\lambda_0}-\Psi_n)\rangle|
                \leq
                \|a(h_t)^*\Psi_{\lambda_0}\| \,
                \|a(l_t)^*(\Psi_{\lambda_0}-\Psi_n)\|
                \\\notag
                &\leq
                (\|h\|_2 + \|h/\sqrt\omega\|_2)
                \|\Psi_{\lambda_0}\|_{(H-b+1)^{1/2}}
                (\|l\|_2 + \|l/\sqrt\omega\|_2)
                \|\Psi_{\lambda_0}-\Psi_n\|_{(H-b+1)^{1/2}},
            \end{align}
            and likewise,
            \begin{align}
                &|\langle a(l_t)\Psi_{\lambda_0},
            a(h_t)(\Psi_{\lambda_0}-\Psi_n)\rangle|
                \leq
                \|a(l_t)\Psi_{\lambda_0}\| \,
                \| a(h_t)(\Psi_{\lambda_0}-\Psi_n)\|
                \\\notag
                &\leq
                (\|l\|_2 + \|l/\sqrt\omega\|_2)
                \|\Psi_{\lambda_0}\|_{(H-b+1)^{1/2}}
                (\|h\|_2 + \|h/\sqrt\omega\|_2)
                \|\Psi_{\lambda_0}-\Psi_n\|_{(H-b+1)^{1/2}},
            \end{align}
            which are ensured by the standard estimates \eqref{eq:st-est}
            and Lemma~\ref{lem:bhh}.
            These bounds allow to take the limit $n\to\infty$ of identity
            \eqref{eq:comm-n} which yields
            \begin{align*}
                \langle a(h_t)^*\Psi_{\lambda_0},
                a(l_t)^*\Psi_{\lambda_0}\rangle
                &=
                \langle a(l_t)\Psi_{\lambda_0},
                a(h_t)\Psi_{\lambda_0}\rangle
                + \langle\Psi_{\lambda_0},\Psi_{\lambda_0}\rangle \, \langle h,
                l\rangle_2
            \end{align*}
            Finally, recalling \eqref{eq:comm-lambda} and
            exploiting (iv) that states $a_\pm(h)\Psi_{\lambda_0}=0$, we find
            \begin{align*}
                \langle a(h)^*_\pm\Psi_{\lambda_0},
                a(l)^*_\pm\Psi_{\lambda_0}\rangle
                &=
                \lim_{t\to\pm\infty}\langle a(h_t)^*\Psi_{\lambda_0}
                a(l_t)^*\Psi_{\lambda_0}\rangle
                =
                \langle\Psi_{\lambda_0},\Psi_{\lambda_0}\rangle \, \langle h,
                l\rangle_2
            \end{align*}
            which concludes the proof of (v).
        \item[(vi)] 
            Let $t\in\R$. Thanks to the standard estimate
            \eqref{eq:st-est}, we find
            \begin{align}
             &   \|a_t(h)(H_f+1)^{-\frac12}\|
                =
                \|e^{itH}a(h_t)(H-b+1)^{-\frac12}e^{-itH}(H-b+1)^{\frac12}(H_f+1)^{-\frac12}\|
                \notag
                \\
                &\leq \|a(h_t)(H-b+1)^{-\frac12}\| \, \|(H-b+1)^{\frac12}(H_f+1)^{-\frac12}\|
                \notag
                \\
                &\leq \|h/\sqrt\omega\|_2\, \|H_f^{\frac12}(H-b+1)^{-\frac12}\|
                \, \|(H-b+1)^{\frac12}(H_f+1)^{-\frac12}\|.
                \label{eq:a-h-hf}
            \end{align}
            Lemma~\ref{lem:bhh} ensures that the right-hand side 
            of \eqref{eq:a-h-hf} is bounded by a finite
            constant $C(h)$ which depends only on $h$. This proves the first
            inequality of (vi). The proof of the second is omitted here as it
            is almost identical.
    \end{enumerate}
\end{proof}

\section{The principle term $T_p(h,l)$}
\label{app:order}

In the section, we prove that    if $G \equiv G(h, l) $ is positive and strictly positive at $\Re \lambda_1 - \lambda_0$ then the absolute of the principal term 
$  T_P(h,l)  $ can be bounded by a strictly positive constant times $g^2$. 
\begin{lemma}
Suppose that  $G \equiv G(h, l) $ is positive and strictly positive at $\Re \lambda_1 - \lambda_0$, then,  for small enough $g$ (depending on $G$), there is a constant $C(h,l)>0$ (independent of $g$) such that 
\begin{align}
\label{claim0}
|T_P(h,l)|\geq  C(h,l) g^2.
\end{align} 
\end{lemma}
\begin{proof}
We set 
\begin{align}
\label{I}
I:=\int \mathrm{d}r      \frac{ G(r) }{  (  r + \lambda_0    - \Re \lambda_1- i g^2 E_1    )  (  r -  \lambda_0 + \overline{\lambda_1}     )  },
\end{align}
and take small enough $g$.
Recalling \eqref{scatteringformulapp1}, we observe  that  
\begin{align}
\label{TP}
T_P(h,l)=g^2 E_1  M I.
\end{align}   
We recall from the discussion below Definition \ref{def:G} that $E_1=E_I +g^{ a }\Delta $,  where  $a>0$, $\Delta\equiv \Delta(g)$ is  uniformly bounded and  $E_I $ is a strictly negative constant that does not depend on $g$, see \eqref{EI}.   
Additionally, it follows from \eqref{projj1}  together with $\norm{\varphi_0\otimes \Omega}=1$ that $\norm{\Psi_{\lambda_0}}\geq C>0$, for some constant $C$ that does not depend on $g$. 
Moreover, we conclude from \eqref{ground1}  that $\Re \lambda_1 -\lambda_0 \geq C>0$ for some constant $C$ (independent of $g$).
Consequently, \eqref{constM} guarantees that there is a constant $C$ (independent of $g$) such that  $|M|\geq C>0$.

This together with \eqref{TP} implies that it suffices to show that there is a constant $C(h,l)>0$ such that 
\begin{align}
\label{claim}
| I |\geq C(h,l),
\end{align}
in order to conclude \eqref{claim0}.

For $\alpha\equiv \alpha_g:=\Re \lambda_1-\lambda_0$ and recalling \eqref{def:e1}, we observe
\begin{align}
\label{I0}
I=\int \mathrm{d}r   \frac{ G(r) }{ (  r -\alpha - i g^2 E_1   )( r +\alpha - i g^2 E_1  )  } 
=\int \mathrm{d}r   \frac{ G(r) \left(   r^2 -\alpha^2 -g^4E_1^2 + 2i g^2 E_1r  \right) }{  ( r^2 -\alpha^2 -g^4E_1^2)^2 +4g^4 E_1^2r^2     }.
\end{align}
Let $c > 0$ be such that $G$ is supported in the complement of the ball or radius $c$
 and center $0$.  Then, we have 
\begin{align}
|\Im (I)|\geq |E_1| \int  \mathrm{d}r G(r)  \frac{ 2 g^2 r  }{  ( r^2 -\alpha^2 -g^4E_1^2)^2 +4g^4 E_1^2c^2     }.
\end{align}
Substituting $s=r^2$, yields 
\begin{align}
|\Im (I)|\geq |E_1| \int \mathrm{d}s G( \sqrt{s} )  \frac{  g^2   }{  ( s -\alpha^2 -g^4E_1^2)^2 +4g^4 E_1^2c^2     }.
\end{align}
Since $G(\alpha) \ne 0$, then for small enough $g$ there is a constant $ r_0  $, that does not depend on $g$  and a constant $ C > 0 $ (independent of $g$) such that $ G(\sqrt{s}) \geq C  $, for every $s \in 
[   \alpha^2 + g^4E_1^2 - r_0,   -\alpha^2 -g^4E_1^2 + r_0     ]$. We apply the change of variables 
$ u = s -  \alpha^2 - g^4E_1^2  $ and obtain
\begin{align}
|\Im (I)|\geq  C |E_1| \int_{- r_0}^{r_0} \mathrm{d}s  \frac{  g^2   }{  s^2 +4g^4 E_1^2c^2     }.
\end{align} 
Finally, we change to the variable $ \tau = s / g^2 $ to obtain:
\begin{align}
|\Im (I)|\geq  C |E_1|  \int_{- r_0/g^2}^{r_0/g^2} \mathrm{d}\tau   \frac{  1   }{  \tau^2 +4 E_1^2c^2     } \geq C |E_1| ,
\end{align}
for small enough $g$ (depending on $G$).  
\end{proof}
\end{appendix}

\section*{List of main notations}
In this section we provide of list of main notations and their place of definition used in this
\begin{center}
    \begin{longtable}{l | l}
    Symbol & Place of definition  \\ \hline
 $E_1$ & below \eqref{eq:lorentzian} \\
 $H_0$, $K$, $H_f$ & \eqref{h0def}\\
 $e_0$, $e_1$ & below \eqref{h0def}\\
  $\omega$ & below \eqref{h0def}\\
   $V$, $\sigma_1$ &  \eqref{interaction}\\
    $f$ &  \eqref{eq:f}\\
     $\mu$ &  \eqref{const:mu}\\
      $H$ &  \eqref{eq:H}\\
       $g$ & below \eqref{eq:H},  see also   Definition \ref{gggg},    \eqref{dorm2}  and Definition 4.3 in \cite{bdh-res}   \\
        $\mathcal H$, $\mathcal K$ &  \eqref{hilbertspace}\\
        $\mathcal F[\mathfrak{h}]$, $\mathfrak{h}$ &  \eqref{fockspace}\\
        $\odot $ & below  \eqref{fockspace}\\
        $\Omega $ &   \eqref{Omega}\\
        $\mathcal F_0 $ & \eqref{fock0}\\
        $\mathit S(\R^3,\C) $ & below  \eqref{fock0}\\
        $a(h) $ &   \eqref{def:a}\\
         $a(h)^* $ &   \eqref{def:a1}\\
         $a(k) $ &   \eqref{eq:aformal}\\
         $a(k)^* $ &   \eqref{eq:akstar}\\
         $\varphi_0 $, $\varphi_1$  &   \eqref{varphi}\\
         $\mathcal D(\bullet)$ & below  \eqref{varphi}\\
         $\sigma(\bullet)$ & below  \eqref{varphi}\\
         $\theta$, $u_\theta$, $U_\theta$ & Definition \ref{complextransf} \\
         $H^\theta$ & \eqref{Hthetaaaa}\\
         $H_f^\theta$, $V^\theta$ & \eqref{eq:Hftheta}\\
         $\omega^\theta$, $f^\theta$ & \eqref{def:thetafncts}\\
         $D(\bullet,\bullet)$ & \eqref{eq:def-disc}\\
         $\lambda_0$, $\lambda_1$ & below Lemma \ref{lemma:specdilh0}\\
         $\Psi_{\lambda_0}$, $\Psi_{\lambda_1}$ & \eqref{eq:gsvec}\\
         $\mathfrak{h}_0$ & \eqref{def:h0}\\
         $a_\pm(h) $ &   \eqref{asymptop}\\
         $a_\pm(h)^* $ &   below \eqref{asymptop} \\
         $\mathcal K^\pm$, $\mathcal H^\pm$ & \eqref{asympthilbert}\\
         $\Omega_\pm$ & \eqref{intertwining}\\
         $S(h,l)$ & \eqref{eq:2bodyscat}\\
         $T(h,l)$ & \eqref{eq:Tmatrix}\\
         $G$ & \eqref{eq:G-def}\\
          $\iota$ &   \eqref{dorm2} and \eqref{dorm222}    \\ 
         $T_P(h,l)$ & \eqref{scatteringformulapp1}\\
         $R(h,l)$ & \eqref{scatteringkernel}\\
         $\nu$ & \eqref{eq:nu}\\
         $\mathcal S$ & \eqref{def:setS}\\
         $\boldsymbol \nu$ & below \eqref{def:setS}\\
         $\rho_0$, $\rho $ & \eqref{rhos}    and \eqref{dorm2}  \\
         $A$ & \eqref{region:A}\\
         $B_0^{(1)}$,  $B_1^{(1)}$ & \eqref{region:Bi1}\\
         $\mathcal C_m(z)$ & \eqref{eq:defcone}\\
         $E_I$ & \eqref{EI}\\
         $\rho_n$ & \eqref{rho22}\\
         $H^{(n),\theta}$ & \eqref{Hntheta}\\
         $H_f^{(n),\theta}$ & \eqref{Hfntheta}\\
         $V^{(n),\theta}$ & \eqref{Vntheta}\\
         $\mathcal H^{(n)}$, $\mathfrak{h}^{(n)}$ & \eqref{hilbertn}\\
         $\tilde H^{(n)}$ & \eqref{Htilde}\\
         $\mathfrak{h}^{(n, \infty)}$ & \eqref{hilbertrest}\\
         $\Omega^{(n, \infty)}$, $P_{\Omega^{(n, \infty)}}$ & below \eqref{hilbertrest}\\
         $\lambda_0^{(n)}$, $\lambda_1^{(n)}$ & above \eqref{projections}\\
          $P_0^{(n),\theta}$, $P_1^{(n),\theta}$ &  \eqref{projections}\\
          $P_0^{\theta}$, $P_1^{\theta}$ &  \eqref{ana}\\
          $\boldsymbol C$ & below    \eqref{dorm222}    \\
          $\mathfrak{F}$, $\mathfrak{F}^{-1}$ & Definition \ref{def:fourier-distri}\\
          $W$ & \eqref{def:W1st}\\
          $\Sigma$ & below \eqref{def:W1st}\\
          $\mathcal H_0$ & \eqref{eq:H0}\\
          $\mathcal F_\text{fin}[\mathfrak{h}_0]$ & \eqref{def:denseF}\\
          $\norm{\bullet}_{\bullet}$ & below  \eqref{def:denseF}\\
          $\Gamma(\epsilon,R)$ & above \eqref{Gamma-parts}\\
          $\Gamma_-(\epsilon,R)$ &  \eqref{Gamma-parts}\\
          $\Gamma_d(R)$ &  \eqref{Gamma-parts}\\
          $\Gamma_c(\epsilon)$ &  \eqref{Gamma-parts}\\
          $\epsilon_n$ & \eqref{eq:epsn}\\
          $R_1(q,Q)$ & \eqref{eq:thl1t}\\
          $P_1(q,Q)$ & \eqref{eq:thl2t}\\
          $\tilde P_1(q,Q)$ & \eqref{eq:thl2ttt}
    \end{longtable}
\end{center}

\section*{Acknowledgement}
D.\ -A.\ Deckert and F.\ H\"anle would like to thank the IIMAS at UNAM and M.\
Ballesteros  the Mathematisches Institut at LMU Munich   for their hospitality. This
project was partially funded by the DFG Grant DE 1474/3-1,  the grants PAPIIT-DGAPA
UNAM  IN108818, SEP-CONACYT 254062, and the junior research group ``Interaction
between Light and Matter'' of the Elite Network Bavaria. M.\  B.\  is a
Fellow of the Sistema Nacional de Investigadores (SNI).
F.\ H.\ gratefully acknowledges financial support by the ``Studienstiftung des deutschen Volkes''
Moreover, the authors express their gratitude for the fruitful discussions with
V.\ Bach, J.\ Faupin, J.\ S.\ M\o ller, A.\ Pizzo and W.\ De Roeck, R. Weder and P. Barberis.

\bibliographystyle{amsplain}
\bibliography{ref}
\end{document}